\documentclass[12pt]{amsart}
\usepackage[latin9]{inputenc}
\usepackage[hmargin=2.5cm,vmargin=2.5cm]{geometry}
\usepackage{amsthm}
\usepackage{amsbsy}
\usepackage{amstext}
\usepackage{amssymb}
\usepackage{graphicx}
\usepackage{esint}

\numberwithin{equation}{section}
\numberwithin{figure}{section}

\newtheorem{theorem}{Theorem}[section]
\newtheorem{lemma}[theorem]{Lemma}
\newtheorem{corollary}[theorem]{Corollary}
\newtheorem{definition}[theorem]{Definition}

\theoremstyle{definition}
\newtheorem{remark}[theorem]{Remark}

\newcommand{\C}{\mathbb{C}}
\newcommand{\R}{\mathbb{R}}
\newcommand{\Z}{\mathbb{Z}}

\newcommand{\G}{\mathbb{G}}
\newcommand{\W}{\mathbb{W}}

\newcommand{\D}{\mathbb{D}}

\newcommand{\alphac}{\alpha^*}
\newcommand\avec{{\boldsymbol{\alpha}}}
\newcommand\avecc{{\boldsymbol{\alpha}^{*}}}
\newcommand\gvec{{\boldsymbol{\gamma}}}
\newcommand\gvecc{{\boldsymbol{\gamma}^{*}}}

\newcommand{\cc}[1]{\overline{#1}}

\newcommand{\RB}[1]{{\bf(To do (RB): #1)}}

\DeclareMathOperator{\im}{Im}

\title[Anomalous nodal count and Dirac points]{Anomalous nodal count and singularities in the dispersion
  relation of honeycomb graphs}
\author{Ram Band}
\address{R.~Band, Department of Mathematics, Technion - Israel Institute of Technology, Haifa 32000, Israel}
\author{Gregory Berkolaiko}
\address{G.~Berkolaiko, Department of Mathematics, Texas A\&M
  University, College Station, TX 77843-3368, USA}
\author{Tracy Weyand}
\address{T.~Weyand, Department of Mathematics, Baylor University,
  Waco, TX 76798-7328, USA}

\begin{document}

\begin{abstract}
  We study the nodal count of the so-called bi-dendral graphs and show
  that it exhibits an anomaly: the nodal surplus is never equal to 0
  or $\beta$, the first Betti number of the graph. According to
  the nodal-magnetic theorem, this means that bands of the magnetic
  spectrum (dispersion relation) of such graphs do not have maxima or
  minima at the ``usual'' symmetry points of the fundamental domain of
  the reciprocal space of magnetic parameters.

  In search of the missing extrema we prove a necessary condition for
  a smooth critical point to happen inside the reciprocal fundamental
  domain. Using this condition, we identify the extrema as the
  singularities in the dispersion relation of the maximal abelian
  cover of the graph (the honeycomb graph being an important example).

  In particular, our results show that the anomalous nodal count is an
  indication of the presence of conical points in the dispersion
  relation of the maximal universal cover. We also discover that the
  conical points are present in the dispersion relation of graphs with
  much less symmetry than was required in previous investigations.
\end{abstract}

\maketitle

%%%%%%%%%%%%%%%%%%%%%%%%%%%%%%%%%%%%%%%%%%%%%%%%
%%%%%%%%%%%%%%%%%%%%%%%%%%%%%%%%%%%%%%%%%%%%%%%%
\section{Introduction}

Quantum graphs and discrete graphs play a significant role in numerous
recent investigations in mathematical physics. On one side, their
importance is apparent in the fields of quantum chaos and spectral
theory \cite{KotSmi_ap99,GnuSmi_ap06,Kuc_incol08,BK_graphs}. On the
other side, they are applied to model various quasi one-dimensional
physical systems such as photonic networks, nanostructures, and
waveguides \cite{Kuc_wrm02,BK_graphs,Post_book12}.  The current work
exposes a curious link between what has so far been considered a more
theoretical aspect of spectral theory on graphs, the nodal count of
eigenfunctions, and a phenomenon with wide-ranging applied
consequences, the existence of Dirac points in the dispersion relation
of periodic structures.

The study of zeros of a quantum graph's eigenfunctions is a fertile area of
research with many questions which are still open.  Some recent
results include bounds on the number of nodal domains
\cite{PokPryObe_mz96,GnuSmiWeb_wrm04,Sch_wrcm06,Ber_cmp08}, specific
formulae for some classes of graphs \cite{BanBerSmi_ahp12}, variational
characterizations \cite{BanBerRazSmi_cmp12,BerWey_ptrsa13}, and inverse
problems \cite{Ban_ptrsa14}. To highlight just one key result, it was
proven in a sequence of works by different authors (see
\cite{Ber_cmp08,BanBerRazSmi_cmp12} and references within) that the number
$\phi_n$ of zeros of the $n$-th eigenfunction obeys the bounds
\begin{equation*}
0\leq\phi_{n} - (n-1)\leq \beta,
\end{equation*}
where $\beta$ is the first Betti number of the graph (intuitively, the
number of cycles).  A natural next question is to consider the
distribution of the values
$\left\{\phi_{n}-(n-1)\right\}_{n=1}^{\infty}$.  One would expect, for example from the magnetic variational characterization of
the ``nodal surplus'' $\phi_{n}-(n-1)$
\cite{Ber_apde13,BerWey_ptrsa13}, that all the integers between $0$ and
$\beta$ appear infinitely often as the nodal surplus of any graph with
Betti number $\beta$.  However, in the current work we show that there
is a family of graphs for which this distribution is not supported on the lower and upper bounds ($0$ and
$\beta$).  Those are the graphs which are obtained as two copies of a tree graph glued together at their corresponding leaves. Hence they are called bi-dendral graphs (see figure~\ref{fig:mandarin}(a) for an example). When the underlying tree is actually a star graph (i.e., a graph with one central vertex connected to all other vertices which are degree one), we call the resulting bi-dendral graph a mandarin graph. We prove that the above nodal count anomaly implies the presence of special singularities in the dispersion relation of the abelian cover version of the mandarin graph. This abelian cover is a periodic infinite graph, also known as the honeycomb graph; it is a tiling of the plane by congruent hexagons whose parallel edges are of equal length (figure~\ref{fig:mandarin}(c)).

Periodic infinite quantum graphs have been fruitfully used to model
diverse physical systems such as photonic crystals
\cite{KucKun_acm02}, graphene \cite{KucPos_cmp07}, and its allotropes
\cite{DoKuc_ns13}.  The Floquet-Bloch theory (see chapter 4 in
\cite{BK_graphs}) reduces the problem of determining the continuous
spectrum of a periodic graph to the study of a parameter-dependent
operator on a compact graph.  The parameters can be interpreted as
magnetic fluxes through the graph; the (now discrete) spectrum as a
function of these parameters is called the dispersion relation.  The
points where two sheets of the dispersion relation touch are of
particular interest, as many physical properties of the material are
related to the location of these points (so-called Dirac points) and
the structure of the bands in their vicinity
\cite{Cas+_rmp09,Nov_nob10,FefWei_jams12, Katsnelson_graphene,FefWei_cmp14}.
The current work characterizes the location of the Dirac points and
the structure of the corresponding eigenstates for periodic graphs,
which are the abelian covers of the mandarin graphs.  This is
particularly important as one member of this family (the 3-mandarin)
is exactly the well-studied hexagonal lattice which models graphene,
although with much reduced symmetry.  The presence of the Dirac points
is the explanation for the anomaly in the nodal count (see
section~\ref{sec:main_results} below); conversely, the anomalous count
is an indication of existence of singularities in the dispersion
relation.  The connection between the two is provided by the
``magnetic-nodal'' theorem \cite{Ber_apde13,CdV_apde13,BerWey_ptrsa13}.

The next section introduces Schr\"{o}dinger operators on quantum
graphs and section~\ref{sec:main_results} presents the main results of
the paper.  Section~\ref{sec:interlacing_and_nodal_count_of_mandarin}
contains the proof of the anomaly of the bi-dendral graph nodal count,
to which we arrive by establishing new eigenvalue interlacing results.
Section~\ref{sec:critical_points_periodic_graphs} studies the
dispersion relation of general periodic graphs and then discusses the
abelian cover of the mandarin graph and makes the connection with its
nodal count.  Finally, in appendix~\ref{sec:discrete_graphs} we adapt
the results of section~\ref{sec:critical_points_periodic_graphs} to
discrete graphs.

%%%%%%%%%%%%%%%%%%%%%%%%%%%%%%%%%%%%%%%%%%%%%%%%%
%%%%%%%%%%%%%%%%%%%%%%%%%%%%%%%%%%%%%%%%%%%%%%%%%
\section{Schr\"{o}dinger operators on quantum graphs}
\label{sec:Schroedinger_operators}

We start by defining a quantum graph, following the notational
conventions of \cite{BK_graphs}, which also contains the proofs of the
background results used in this section.

Let $\Gamma$ be a compact metric graph with vertex set $V$ and edge
set $E$. Let $\widetilde{H}^{k}(\Gamma,\C)$ be the space of all
complex-valued functions that are in the Sobolev space $H^{k}(e)$
for each edge, or in other words
\begin{equation*}
\widetilde{H}^{k}(\Gamma,\C)=\bigoplus_{e\in E}H^{k}(e).
\end{equation*}
 Consider the Schr\"odinger operator with electric potential $q:\Gamma\rightarrow\R$
defined by
\begin{equation*}
H^{0}:f\mapsto-\frac{d^{2}f}{dx^{2}}+qf,
\end{equation*}
acting on the functions from $\widetilde{H}^{2}(\Gamma,\C)$ satisfying
the $\delta$-type matching conditions
\begin{equation}
\left\{ \begin{array}{l}
f(x)\mbox{ is continuous at $v$},\\
\displaystyle\sum_{e\in{E_{v}}}\frac{df}{dx_{e}}(v)=\chi_{v}f(v),\qquad\chi_{v}\in\R
\end{array}\right.\label{eq:vconditions}
\end{equation}
at all vertices $v \in V$. Here the potential $q(x)$ is assumed to be piecewise continuous.
The set $E_{v}$ is the set of edges joined at the vertex $v$; by
convention, each derivative at a vertex is taken into the corresponding
edge. We denote by $x_{e}$ the local coordinate on the edge $e$.
On vertices of degree more than one we usually take $\chi_{v}=0$,
the so-called Neumann condition. Non-zero $\chi_{v}$ will be used
on vertices of degree one, where we also allow the Dirichlet condition
$f(v)=0$, which is formally equivalent to $\chi_{v}=\infty$. We note that a Neumann vertex of degree two can be absorbed into its neighboring edges unifying them both to a single edge (whose length equals the sum of both), without changing the graph spectral properties. At vertices
of degree two we will sometimes be using the so-called anti-Neumann vertex
condition. Namely, for a vertex $v$ of degree two, which is connected
to the edges $e_{+}$ and $e_{-}$, the anti-Neumann condition is
\begin{equation}
  \label{eq:antiN_def}
  \left.f\right|_{e_{-}}(v) = -\left.f\right|_{e_{+}}(v),
  \qquad
  \left.f'\right|_{e_{-}}(v) = \left.f'\right|_{e_{+}}(v),
\end{equation}
where the direction of the derivative is taken from the vertex $v$
into the edge $e_{\pm}$. Note that the anti-Neumann condition does
not fall in the class of vertex conditions \eqref{eq:vconditions}.

The operator $H^{0}$ is self-adjoint, bounded from below, and has
a discrete set of eigenvalues that can be ordered as
\begin{equation*}
  \lambda_{1} \leq \lambda_{2}\leq\ldots\leq\lambda_{n}\leq\ldots\hspace{0.1in}.
\end{equation*}
Throughout the paper we will frequently assume that the graph's eigenvalues
and eigenfunctions are generic in the following sense:

\begin{definition}\label{def:genericity}
~
\begin{enumerate}
\item The eigenpair $\left(\lambda,f\right)$ is called generic if the eigenvalue
$\lambda$ is simple and the corresponding eigenfunction $f$ is
different than zero on every vertex.
\item Both $\lambda$ and $f$ in a generic pair are also called generic.
\item A quantum graph is generic if all of its eigenpairs are generic.
\end{enumerate}
\end{definition}

The conditions which guarantee the graph's genericity are discussed in
\cite{Fri_ijm05,BerLiu_prep14}.  We will count zeros only for generic
eigenfunctions.

\begin{definition}\label{def:nodal_count}
~
\begin{enumerate}
\item Let $f_{n}$ be a generic eigenfunction corresponding to the $n$-th eigenvalue, $\lambda_n$ (counting with multiplicities). Denote by $\phi_{n}$ the
number of its internal zeros. Namely, we do not include the Dirichlet
vertices, if they exist, in the count.
\item The quantity $\sigma_{n}:=\phi_{n}-\left(n-1\right)$ is called the
nodal surplus.
\end{enumerate}
\end{definition}

It was recently discovered that the graph's nodal count is closely
related to properties of the magnetic Schr\"odinger operator on the graph. This connection
is described in theorem~\ref{thm:nodal_mag} below. The magnetic Schr\"odinger
operator on $\Gamma$ is given by
\begin{equation*}
H^{A}(\Gamma):f\mapsto-\left(\frac{d}{dx}-iA(x)\right)^{2}f+qf,\qquad f\in\widetilde{H}^{2}(\Gamma,\C),
\end{equation*}
 where the magnetic potential, $A(x)$, is a one-form (namely, the
sign of $A(x)$ changes with the orientation of the edge). The $\delta$-type
boundary conditions are now modified to the following at all vertices $v \in V$:
\begin{equation*}
\left\{ \begin{array}{l}
f(x)\mbox{ is continuous at $v$},\\
\displaystyle\sum_{e\in{E_{v}}}\left(\frac{df}{dx_{e}}(v)-iA(v)f(v)\right)=\chi_{v}f(v),\qquad\chi_{v}\in\R.
\end{array}\right.
\end{equation*}

Let $\beta=|E|-|V|+1$ be the first Betti number of the graph $\Gamma$,
i.e.\ the rank of the graph's fundamental group. Informally
speaking, $\beta$ is the number of ``independent'' cycles on the
graph and hence is zero if the graph is a tree. Up to a change of gauge, a magnetic field on a graph is fully
specified by $\beta$ fluxes $\alpha_{1},\alpha_{2},\ldots,\alpha_{\beta}$,
defined as
\begin{equation*}
\alpha_{j}=\oint_{\tau_{j}}A(x)dx\mod2\pi,
\end{equation*}
where $\{\tau_{j}\}$ is a set of generators of the fundamental group.
In other words, magnetic Schr\"odinger operators with different magnetic
potentials $A(x)$, but the same fluxes
$(\alpha_{1},\ldots,\alpha_{\beta})$, are unitarily
equivalent. Therefore, the eigenvalues $\lambda_{n}(H^{A})$ can be
viewed as functions of $\avec=(\alpha_{1},\ldots,\alpha_{\beta})$.
The connection between this function and the nodal count is explicated
in the following theorem.

\begin{theorem} \label{thm:nodal_mag} Let $\left(\lambda_{n},f_{n}\right)$
be a generic eigenpair of $H^{A}$ with $A$ corresponding to some
flux $\avecc\in\left\{ 0,\pi\right\} ^{\beta}$. Then $\avecc$
is a non-degenerate critical point of the function $\lambda_{n}(\avec)$
and its Morse index is equal to the nodal surplus, $\sigma_{n}$.
\end{theorem}

\begin{remark}
  \label{rem:nodal_mag}
  The above theorem is proved in \cite{BerWey_ptrsa13} for
  $\avec=(0,\ldots,0)$, extending to quantum graphs an earlier result
  on discrete graphs \cite{Ber_apde13} (also see \cite{CdV_apde13}).
  One may repeat the same proof from \cite{BerWey_ptrsa13} for all other $\avec\in\left\{
    0,\pi\right\} ^{\beta}$, noticing the following:
\begin{enumerate}
\item Introducing a vertex of degree two with an anti-Neumann vertex
  condition on some edge results in an operator which is unitarily
  equivalent to the original operator with a magnetic potential on
  this edge integrating to $\pi$.
\item The anti-Neumann condition keeps the operator self-adjoint and
  does not change either of the quadratic forms used in section 4 of
  \cite{BerWey_ptrsa13}; the only change needed is the additional
  condition $f\big|_{e_{-}}(v) = -f\big|_{e_{+}}(v)$ imposed on the
  domain of the quadratic form.
\item \label{itm:tree_count}
  The $n^{\textrm{th }}$ eigenfunction of a tree graph has $n-1$
  internal zeros and this holds even if some anti-Neumann conditions
  are imposed on the graph (relevant for the proof of theorem 3.3,
  part 1 in \cite{BerWey_ptrsa13}).
\end{enumerate}
\end{remark}

As a corollary, we get an earlier result (see \cite{Ber_cmp08,BanBerRazSmi_cmp12}
and references therein),
\begin{equation}
0\leq\sigma_{n}\leq\beta.\label{eq:nodal_surplus_bounds}
\end{equation}
The theorem above makes one wonder whether the whole range of integers
from $0$ to $\beta$ is covered by $\sigma_{n}$ for any graph. We will
show that there is a family of quantum graphs for which the answer
to this question is negative. These graphs, to be discussed next,
and their spectral properties are the focus of this paper.

%%%%%%%%%%%%%%%%%%%%%%%%%%%%%%%%%%%%%%%%%%%%%%%%%%%%%
%%%%%%%%%%%%%%%%%%%%%%%%%%%%%%%%%%%%%%%%%%%%%%%%%%%%%
\section{Main results\label{sec:main_results}}

\subsection{Anomalous nodal count}

\begin{figure}[t]
  \centering
  \includegraphics[scale=0.8]{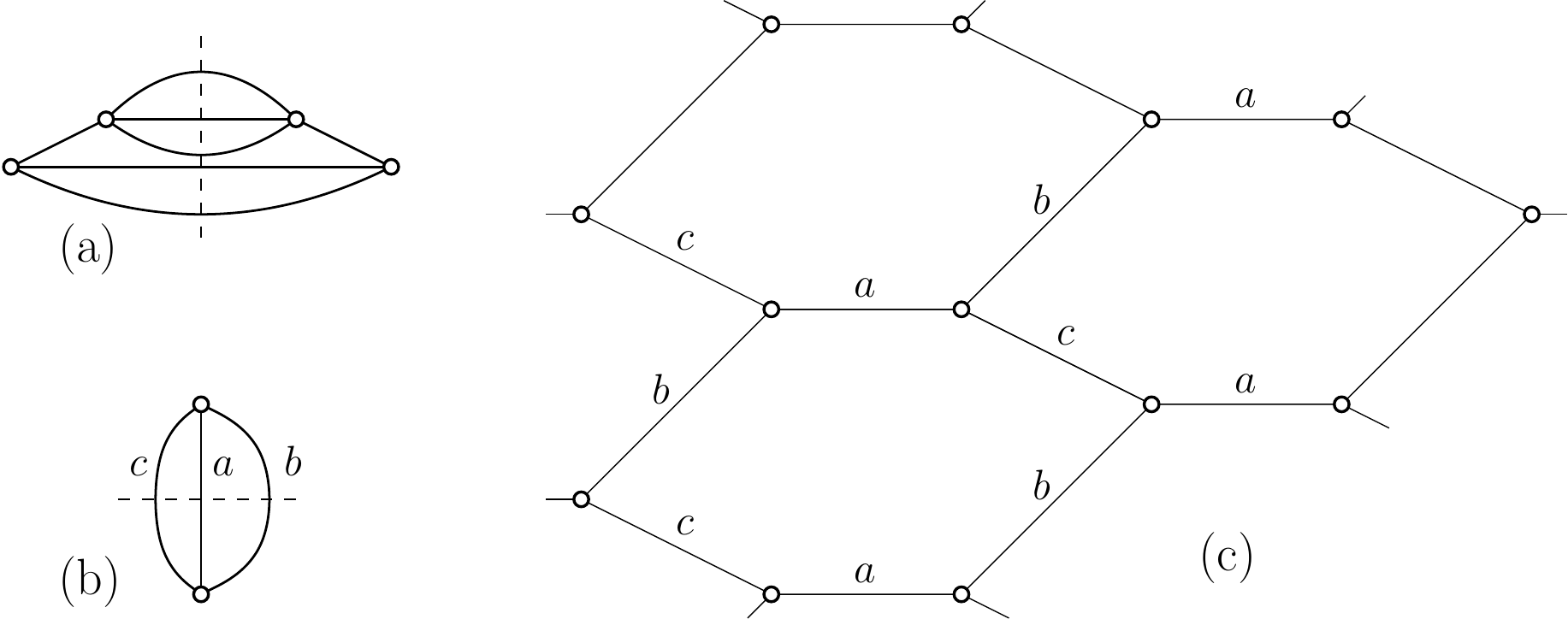}
  \caption{An example of a bi-dendral graph (a), the symmetry axis is
    indicated by the dashed line; a 3-mandarin graph (b) and its
    maximal abelian cover (c). The abelian cover graph has honeycomb
    structure although less symmetry than a regular honeycomb lattice:
    only the parallel edges have the same length.}
\label{fig:mandarin}
\end{figure}

Consider two copies of a tree graph with $d$ leaves (vertices of degree one) which are glued
together by identifying corresponding leaves of both. We impose
Neumann conditions at all vertices and call the resulting quantum graph
a \emph{bi-dendral graph}. When the underlying tree graph is
chosen to be a star graph, we obtain a graph consisting of two
vertices and $d$ edges connecting them (recalling that Neumann
vertices of degree two can be absorbed into an edge). We will call
such a graph a \emph{$d$-mandarin graph} (see
figure~\ref{fig:mandarin}).  A bi-dendral graph has an obvious
symmetry axis which passes through all points arising from gluing the
leaves. We call the edges which cross this symmetry axis the
\emph{middle edges}. We extend the definition of a bi-dendral graph by
allowing for vertices of degree two with anti-Neumann condition to be
present on those middle edges. Finally, we will be assuming the graph
edge lengths are chosen such that either a particular eigenvalue is
generic or the graph is generic (see definition~\ref{def:genericity}).

\begin{theorem} \label{thm:mandarin_mag_surplus_bounds} Let $\Gamma$
be a bi-dendral graph with $d$ middle edges, such that $a\geq 0$ of them have
anti-Neumann conditions imposed at an intermediate point.
If the $n^{\textrm{th}}$ eigenpair is generic and $n>1$ or $a>0$,
then the nodal surplus $\sigma_{n}$ satisfies
\begin{equation}
1\leq\sigma_{n}\leq\beta-1=d-2.\label{eq:mandarin_mag_surplus_bounds}
\end{equation}
\end{theorem}

\begin{remark}
Compare \eqref{eq:mandarin_mag_surplus_bounds} with the general bounds of the nodal surplus, \eqref{eq:nodal_surplus_bounds}.
\end{remark}
Consider the eigenvalue $\lambda_{n}(\avec)$ as a function of $\avec=(\alpha_{1},\ldots,\alpha_{\beta})$,
the total fluxes through the cycles of the graph (for some pre-fixed
basis of the fundamental group). It is well-known that the points
$\avec\in\left\{ 0,\pi\right\} ^{\beta}$ are critical points of the
function $\lambda_{n}(\avec)$: they are the fixed points of the symmetry
transformation $\avec\mapsto-\avec\,\left(\textrm{mod }2\pi\right)$.
Henceforth we will refer to these points as \emph{symmetry points}.

Theorem~\ref{thm:nodal_mag} states that the nodal surplus of the
eigenfunction at a symmetry point is equal to the Morse index of $\lambda_{n}(\avec)$
at this point. We now infer from theorem~\ref{thm:mandarin_mag_surplus_bounds}
that \emph{none} of the symmetry points are extrema of the eigenvalues
$\lambda_{n}(\avec)$ (with the exception of the eigenvalue $\lambda_{1}(0,\ldots,0)$,
which is always a minimum). Since the $\lambda_{n}(\avec)$ are continuous
functions on the compact torus $(-\pi,\pi]^{\beta}$, the extrema
must be achieved somewhere else.

%%%%%%%%%%%%
\subsection{A condition for an internal critical point}

Searching for the ``missing'' extremal points of the functions
$\lambda_{n}(\avec)$, we need to look in the bulk of the fundamental
domain of the $\avec$-space.  Critical (extremal) points that are not
one of the previously identified symmetry points will be called
\emph{internal critical (extremal) points}. A common (and mistaken)
assumption is that there are no internal extremal points; it was
pointed out in \cite{HarKucSob_jpa07} that such points may occur and explicit examples were constructed in \cite{HarKucSob_jpa07, ExnKucWin_jpa10}.
Moreover, in the case of mandarin graphs their occurrence is simply
unavoidable; the extrema must be achieved somewhere.

Our search for extremal points is aided by proving a necessary condition
that simple internal critical points must obey.

\begin{theorem} \label{thm:magnetic_cp_condition} Let $\Gamma$ be a graph with first Betti number equal to $\beta$ and denote by $\avec\in(-\pi,\pi]^{\beta}$ the total fluxes through some choice of cycles of the graph that form a
basis of its fundamental group. If the eigenvalue $\lambda_{n}(\avec)$ has an internal critical point $\avecc$ (i.e., $\avecc\notin\left\{ 0,\pi\right\} ^{\beta}$)
and $\lambda_{n}(\avecc)$ is simple, then the eigenfunction corresponding
to $\lambda_{n}(\avecc)$ is equal to zero at some vertex of the graph.
\end{theorem}

%%%%%%%%%%%
\subsection{The missing extrema are found at touching bands}

The necessary condition given in
theorem~\ref{thm:magnetic_cp_condition} allows us to show that an
internal critical point of a mandarin graph must be
degenerate. Calling the graph of the function $\lambda_{n}(\avec)$ a
band, we get the following:

\begin{theorem}
  \label{thm:mandarin_magnetic_cp}
  For a generic mandarin graph, all extrema of $\lambda_{n}(\avec)$,
  apart from the minimum of $\lambda_{1}(\avec)$, are achieved at points
  where two bands touch.
\end{theorem}

\begin{figure}
  \centering
  \includegraphics[scale=0.4]{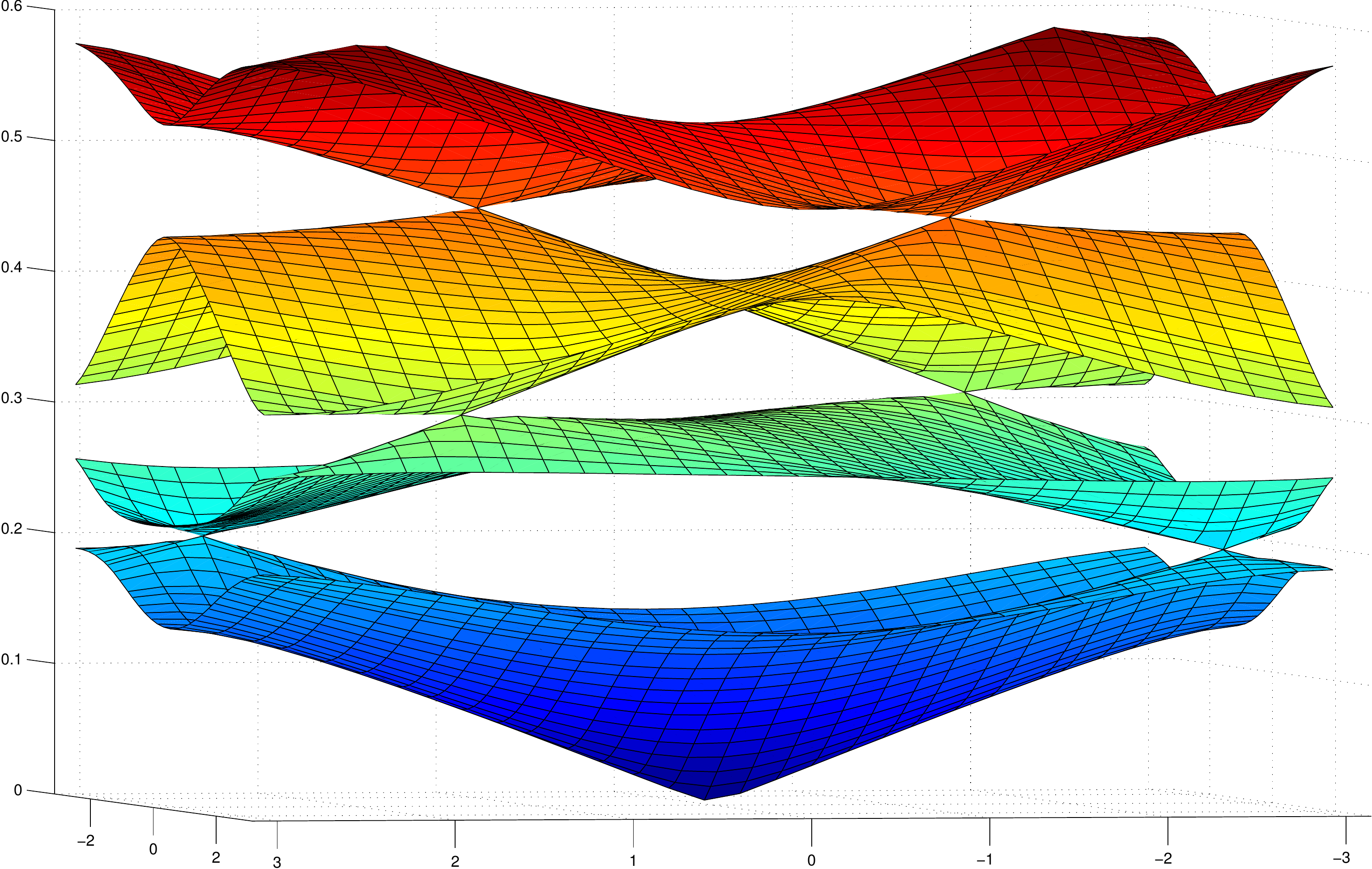}
  \caption{First four bands of a 3-mandarin graph, as functions of
    $\beta=2$ magnetic fluxes.  The vertical axis is $k =
    \sqrt{\lambda}$; in this scale the bands have comparable sizes
    \cite{BanBer_prl13}.  This choice of scale also explains the
    conical point at the bottom of the lowest band (in contrast,
    $\lambda_1(\avec)$ is differentiable at $\boldsymbol{0}$).}
  \label{fig:mandarin_disp}
\end{figure}

Figure~\ref{fig:mandarin_disp} illustrates the result of the theorem:
each pair of consecutive bands touch at two symmetrically located points.

It is well-known that the magnetic spectrum of a graph as a function
of the fluxes coincides, by Floquet--Bloch reduction, with the
dispersion relation of the abelian cover of the graph.  Here the
phases $\avec$ take the meaning of the quasi-momenta of the
Bloch functions. For a 3-mandarin graph there are two
$\alpha$-parameters, and a degeneracy in two parameters is typically a
conical point \cite[app.~10]{Arnold_mechanics}.  The maximal abelian
cover of a 3-mandarin is a honeycomb-like graph, where only the
parallel edges are required to have equal length (see
figure~\ref{fig:mandarin}).  The above chain of results shows that
typically there are conical (``Dirac'') points in the dispersion
relation of such a graph between each pair of adjacent bands.

%%%%%%%%%%%%%%%%%%%%%%%%%%%%%%%%%%%%%%%%%%%%%%%%%%%
%%%%%%%%%%%%%%%%%%%%%%%%%%%%%%%%%%%%%%%%%%%%%%%%%%%
\section{Interlacing and nodal count for mandarin graphs}
\label{sec:interlacing_and_nodal_count_of_mandarin}

%%%%%%%%%
\subsection{Eigenvalue interlacing}

We begin working towards the proof of
theorem~\ref{thm:mandarin_mag_surplus_bounds} by establishing some
eigenvalue interlacing results, which are heavily based on the
inequalities of \cite[sect.~3.1]{BK_graphs}.  The following lemma was
proved as part of lemma 4.3 of \cite{BanBerSmi_ahp12}.

\begin{lemma}
  \label{lem:separated_graph}
  Let $\Gamma$ be a compact connected graph and let $\Lambda = \lambda_n(\Gamma)$ be a simple eigenvalue of $\Gamma$ with the corresponding eigenfunction $f$. Let $\Gamma_c$ be a graph which is the union of $k+1$ connected components such that the following are true:
    \begin{enumerate}
      \item $\Gamma$ can be obtained from $\Gamma_c$ by $k$ operations of gluing a pair of vertices together and adding the parameters $\chi$ of their $\delta$-type conditions.
      \item The function $f$ restricted to any of the components of $\Gamma_c$ is an eigenfunction of that component.
    \end{enumerate}
   Then
  \begin{align}
    \label{eq:deg_eig_components}
    &\lambda_{n-1}(\Gamma_c) < \Lambda < \lambda_{n+k+1}(\Gamma_c),\\
    &\lambda_n(\Gamma_c) = \lambda_{n+1}(\Gamma_c)
    = \ldots = \lambda_{n+k}(\Gamma_c) = \Lambda.
  \end{align}
\end{lemma}

The next lemma we formulate in somewhat greater generality than that which will be
necessary in subsequent derivations.

\begin{lemma}
  \label{lem:interlacing_maxD}
  Let $\Gamma$ be a graph with a Neumann vertex $v$ of degree $d$
  whose removal separates the graph into $d$ disjoint subgraphs.  We
  denote its edge set by $E_{v}$.  Let $r\leq d$ be a non-negative integer. For a subset $E_{D}$ of $E_{v}$, with $\left|E_{D}\right|=r$,
  define $\Gamma_{E_{D}}$ to be the modification of the graph $\Gamma$
  obtained by imposing the Dirichlet condition at $v$ for edges from
  $E_{D}$ and leaving the edges from $E_{v}\setminus E_{D}$ connected
  at $v$ with the Neumann condition (see
  figure~\ref{fig:max_dirichlet} for an example).  Then
  \begin{equation}
    \lambda_{n-1}(\Gamma)
    \leq \min_{|E_{D}|=r}\lambda_{n}(\Gamma_{E_{D}})
    \leq \lambda_{n}(\Gamma)
    \leq \max_{|E_{D}|=r}\lambda_{n}(\Gamma_{E_{D}})
    \leq \lambda_{n+1}(\Gamma).
    \label{eq:max_Dirichlet}
  \end{equation}
\end{lemma}

\begin{figure}
  \includegraphics[scale=0.8]{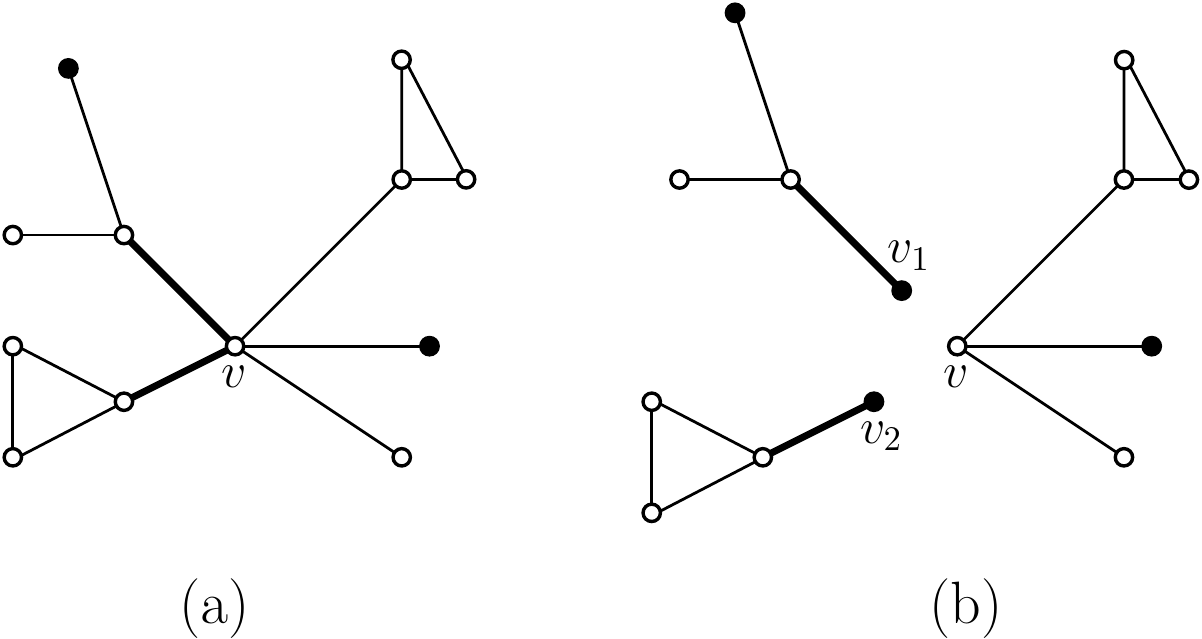}
  \caption{The original graph $\Gamma$, (a), and the resulting graph $\Gamma_{E_{D}}$,
    (b), when the set $E_{D}$ comprises of the $r=2$ edges joining $v$
    from the left (shown in thicker lines). Vertices with Neumann conditions
    are shown as empty circles, and Dirichlet conditions are indicated by
    filled circles.}
  \label{fig:max_dirichlet}
\end{figure}

\begin{proof}
  Without loss of generality we will assume that the eigenvalue
  $\lambda_n(\Gamma)$ is simple.  Indeed, one can resolve multiplicity
  by an arbitrarily small perturbation of the edge lengths
  \cite{Fri_ijm05}.  Since the $n$-th eigenvalue is a continuous
  function of the lengths \cite[thm 3.1.11]{BK_graphs}, if the
  (non-strict) inequality is true for simple eigenvalues, it remains
  true when passing to the limit of zero perturbation.

  \emph{The outside inequalities:} Start from the graph $\Gamma$ and
  disconnect any $r$ edges keeping them joined together with the
  Neumann condition and denoting the resulting graph by $\Gamma'$.
  From \cite[thm 3.1.11]{BK_graphs}  we get
  \begin{equation*}
    \lambda_{n-1}\left(\Gamma\right)
    \leq \lambda_{n}\left(\Gamma'\right)
    \leq \lambda_{n+1}\left(\Gamma'\right)
    \leq \lambda_{n+1}\left(\Gamma\right).
  \end{equation*}
  Now, changing the vertex condition from Neumann to Dirichlet
  disconnects the $r$ edges and results in
  $\lambda_{n}\left(\Gamma'\right) \leq
  \lambda_{n}\left(\Gamma_{E_{D}}\right)\leq\lambda_{n+1}\left(\Gamma'\right)$
  \cite[thm 3.1.8]{BK_graphs}.

  \emph{The lower bound for the maximum:} We describe the choice of $E_{D}$ that will
  fulfill the inequality. We start with the the $n$-th eigenfunction
  $f$ of $\Gamma$ (normalizing so that $f(v)>0$).  We now choose the $r$ edges
  $e$ with the largest values of $f_{e}'(v)$ to be in $E_{D}$ and get
  \begin{equation*}
    \sum_{e\in E\setminus E_{D}}f_{e}'(v)\leq0.
  \end{equation*}
  Consider $f$ on the graph $\Gamma_{E_{D}}'$ obtained by
  disconnecting the edges from the set $E_{D}$ at the vertex $v$. We
  supply $\Gamma_{E_{D}}'$ with conditions that ensure that $f$ is an
  eigenfunction. Namely, the function $f$ satisfies the $\delta$-type
  condition at the vertex $v$ with some coefficient
  $\chi_{v}\leq0$. At the new vertices $v_{1}$, $v_{2}$ etc. (see
  figure~\ref{fig:max_dirichlet}) some $\delta$-type conditions are
  also satisfied. The sum of the parameters of those $\delta$-type
  conditions together with $\chi_{v}$ is zero. From
  lemma~\ref{lem:separated_graph} we have
  \begin{equation*}
  \lambda_{n}(\Gamma_{E_{D}}')=\lambda_{n}(\Gamma).
  \end{equation*}
  Now we modify the conditions of $\Gamma_{E_{D}}'$, increasing
  $\chi_v$ to zero and the $\delta$-type parameters at the new
  vertices of degree one to $\infty$ (Dirichlet condition), obtaining
  the graph $\Gamma_{E_{D}}$.  All these operations increase the
  eigenvalues \cite[thm 3.1.8]{BK_graphs}; in particular,
  \begin{equation*}
  \lambda_{n}(\Gamma_{E_{D}})\geq\lambda_{n}(\Gamma_{E_{D}}')=\lambda_{n}(\Gamma).
  \end{equation*}

  \emph{The upper bound for the minimum:} 
  In this case, we choose the $r$ edges
  $e$ with the smallest values of $f_{e}'(v)$ to be in $E_{D}$ and get
  \begin{equation*}
    \sum_{e\in E\setminus E_{D}}f_{e}'(v) \geq 0,
  \end{equation*}
  and, therefore, $\chi_{v}\geq0$.  From
  lemma~\ref{lem:separated_graph} we have
  \begin{equation*}
    \lambda_{n+r}(\Gamma_{E_{D}}')=\lambda_{n}(\Gamma).
  \end{equation*}
  Now we \emph{decrease} $\chi_v$ to 0, resulting in
  \begin{equation*}
    \lambda_{n+r}(\Gamma_{E_{D}}'') \leq \lambda_{n+r}(\Gamma_{E_{D}}'),
  \end{equation*}
  and then increase the parameters of the new vertices of degree one to
  $\infty$, obtaining
  \begin{equation*}
    \lambda_{n}(\Gamma_{E_{D}}) \leq \lambda_{n+r}(\Gamma_{E_{D}}'').
  \end{equation*}
  The desired result is an obvious chaining of the obtained inequalities.
\end{proof}

\begin{remark} The condition in lemma~\ref{lem:interlacing_maxD} that removing a vertex of degree $d$ separates the graph into $d$ disjoint subgraphs is necessary for the lower bound in \eqref{eq:max_Dirichlet}. This is used in the proof above when applying lemma~\ref{lem:separated_graph}, and indeed there exist examples where the bound does not hold when the condition is not satisfied.
\end{remark}

\begin{remark}
  In the case of a graph which is ``spherically-symmetric'' around
  vertex $v$ (i.e. the $d$ disjoint subgraphs obtained by removing
  the vertex $v$ are identical), we immediately get
  \begin{equation*}
    \lambda_{n}(\Gamma_{E_{D}}) = \lambda_{n}(\Gamma),
  \end{equation*}
  for all $n$ and $E_D$.  A special case of this is proved and used in a
  recent work of Demirel-Frank \cite{Dem_prep15}.
\end{remark}

%   \begin{figure}[t]
%     \centering
%     \includegraphics[scale=0.75]{DNstar_mixed}
%     \caption{The sequence of graph modifications resulting in inequality
%       \eqref{eq:DNstar_mixed}.  The transition from graph $\Gamma'$ to
%       graph $\Gamma''$ results in the shift of the eigenvalue by two
%       places up (changing two top left intervals) and one place down
%       (changing the bottom interval).}
%     \label{fig:DNstar_mixed}
%   \end{figure}

\begin{lemma} \label{lem:DNtree_mixed} Let $\Gamma_{DN}$ be a tree
graph with zero potential, Dirichlet conditions at $t$ of its
leaves, and Neumann conditions at all of its other leaves and at all
of its internal vertices. Let $\Gamma_{ND}$ be the same graph with
Dirichlet conditions changed to Neumann and vice versa at all leaves.
If $k+t-1>0$, then
\begin{equation}
\lambda_{k}(\Gamma_{DN})\leq\lambda_{k+t-1}(\Gamma_{ND}).\label{eq:DNtree_mixed}
\end{equation}
\end{lemma}

\begin{remark} Naively applying the standard interlacing results
(see, e.g. \cite[thm 3.1.8]{BK_graphs}) while changing the condition at
each leaf, we would get the weaker result
\begin{equation*}
\lambda_{k}(\Gamma_{DN})\leq\lambda_{k+t}(\Gamma_{ND}).
\end{equation*}
\end{remark}

\begin{remark} Taking the inequality above with $t=0$ gives an inequality
similar to that known for domains in $\R^{d}$. In general, if $\Omega\subset\R^{d}$
is a domain whose Neumann spectrum is discrete, then $\lambda_{k}^{\left(N\right)}(\Omega)\leq\lambda_{k-1}^{\left(D\right)}(\Omega)$
where the superscript $\left(N\right)\backslash\left(D\right)$ denotes
the Neumann$\backslash$Dirichlet spectrum \cite{Friedlander_rma91,Filonov_aia04}.
\end{remark}

\begin{proof}

We start by introducing the spectral counting function
\[
N\left(\Gamma;\lambda\right):=\left|\left\{ n\left|\lambda_{n}\left(\Gamma\right)\leq\lambda\right.\right\} \right|,
\]
which we will use to rephrase eigenvalue inequalities. In particular,
we have the following equivalence:
\begin{equation}
  \lambda_k\left(\Gamma_1\right) \leq
  \lambda_{k+s}\left(\Gamma_2\right)
  \mbox{ for all }k
  \qquad \Leftrightarrow \qquad
  N\left(\Gamma_2;\lambda\right) \leq N\left(\Gamma_1;\lambda\right) + s.
  \label{eq:DNtree_counting_func_equivalence}
\end{equation}
To go from left to right we observe that if $N(\Gamma_1;\lambda) = n$,
then $\lambda < \lambda_{n+1}(\Gamma_1) < \lambda_{n+s+1}(\Gamma_2)$,
by definition and left inequality, correspondingly.  Applying the
definition of $N$ once again, we get $N\left(\Gamma_2;\lambda\right)
\leq n+s$.  The left follows from the right by the substitution
$\lambda=\lambda_{k+s}\left(\Gamma_{2}\right)$.  Using the equivalence
above, we write \eqref{eq:DNtree_mixed} as
\begin{equation}
  N\left(\Gamma_{ND};\lambda\right) \leq N\left(\Gamma_{DN};\lambda\right)+t-1,
  \label{eq:DNtree_lem_phrased_as_counting}
\end{equation}
which is what we now prove using induction on the number of internal
vertices of $\Gamma_{DN}$. The starting point is to notice
that the lemma holds for intervals (no internal vertices) with either
Dirichlet or Neumann conditions at each of their endpoints. Denoting
intervals with different types of boundary conditions by $I_{DD},\, I_{NN},$ and $I_{DN}=I_{ND}$,
we have
\begin{align*}
  \lambda_{k}\left(I_{DD}\right)
  &= \lambda_{k+1}\left(I_{NN}\right)
  = \left(\frac{\pi}{l}k\right)^{2}\\
  \lambda_{k}\left(I_{NN}\right)
  &=\lambda_{k-1}\left(I_{DD}\right)
  = \left(\frac{\pi}{l}\left(k-1\right)\right)^{2}\\
  \lambda_{k}\left(I_{DN}\right)
  &=\lambda_{k}\left(I_{ND}\right)
  = \left(\frac{\pi}{2l}\left(2k-1\right)\right)^{2},
\end{align*}
where $l$ is the interval length. Therefore the lemma holds for an
interval, and for such graphs we even get an equality in \eqref{eq:DNtree_mixed}.
Assume that the lemma holds for all tree graphs which have no more
than $M$ internal vertices. Let $\Gamma_{DN}$ be a tree graph which
satisfies the lemma's conditions and has $M+1$ internal vertices.
Let $v$ be an internal vertex of $\Gamma_{DN}$ with degree $d$.
Note that $v$ has a Neumann condition as it is an internal
vertex. We start with $\lambda_{k}\left(\Gamma_{DN}\right)$ and apply
lemma~\ref{lem:interlacing_maxD} with $r=d-2$ on the vertex $v$
(choosing $E_{D}$ which results in the maximal eigenvalue). Denoting
the resulting graph by $\Gamma'$, we have from lemma~\ref{lem:interlacing_maxD} that
\begin{equation}
\lambda_{k}\left(\Gamma_{DN}\right)\leq\lambda_{k}\left(\Gamma'\right),\label{eq:}
\end{equation}
or equivalently from \eqref{eq:DNtree_counting_func_equivalence},
\begin{equation}
  N\left(\Gamma';\lambda\right) \leq N\left(\Gamma_{DN};\lambda\right).\label{eq:DNtree_ineq1}
\end{equation}
 The resulting graph, $\Gamma'$, is a union of $d-1$ tree graphs,
$\Gamma'=\cup_{i=1}^{d-1}\Gamma_{DN}^{\left(i\right)}$. Note that
$d-2$ of these graphs have a Dirichlet leaf vertex which originates
from $v$, and therefore these graphs have less internal vertices
than $\Gamma_{DN}$. In one graph out of the $\Gamma_{DN}^{\left(i\right)}$,
however, $v$ becomes a Neumann vertex of degree $2$ and can be
absorbed into the edge. Therefore this graph
also has strictly less internal vertices than $\Gamma_{DN}$.  Denote
by $t_{i}$ the number of Dirichlet leaves of $\Gamma_{DN}^{\left(i\right)}$
and note that
\begin{equation*}
  \sum_{i=1}^{d-1} t_i = t + d - 2,
\end{equation*}
since we added $d-2$ Dirichlet vertices by splitting the vertex $v$.

We get by the induction assumption
\begin{equation}
  N\left(\Gamma_{ND}^{\left(i\right)};\lambda\right)
  \leq N\left(\Gamma_{DN}^{\left(i\right)};\lambda\right)+t_{i}-1,
  \label{eq:DNtree_ineq2}
\end{equation}
where $\Gamma_{ND}^{\left(i\right)}$ is the same as
$\Gamma_{DN}^{\left(i\right)}$, but with Dirichlet and Neumann
conditions at all leaves interchanged.

Denoting $\Gamma''=\cup_{i=1}^{d-1} \Gamma_{ND}^{(i)}$
and using the additivity of the spectral counting function we get
\begin{equation}
  N\left(\Gamma'';\lambda\right)
  = \sum_{i=1}^{d-1}N\left(\Gamma_{ND}^{(i)};\lambda\right)
  \leq
  \sum_{i=1}^{d-1}\left(N\left(\Gamma_{DN}^{(i)};\lambda\right)+t_i-1\right)
  = N\left(\Gamma';\lambda\right)+t-1.
  \label{eq:DNtree_ineq3}
\end{equation}
All that remains is to note that $\Gamma_{ND}$ is
obtained from $\Gamma''$ by gluing all vertices which originated
from $v$ and are equipped (in $\Gamma''$) with Neumann conditions;
such gluing increases the eigenvalue \cite[thm 3.1.11]{BK_graphs}, so
$\lambda_{k+t-1}\left(\Gamma''\right)\leq\lambda_{k+t-1}\left(\Gamma_{ND}\right)$
and therefore
\begin{equation}
  N\left(\Gamma_{ND};\lambda\right)
  \leq N\left(\Gamma'';\lambda\right).
  \label{eq:DNtree_ineq4}
\end{equation}
Combining inequalities \eqref{eq:DNtree_ineq1},\eqref{eq:DNtree_ineq3},\eqref{eq:DNtree_ineq4},
we get \eqref{eq:DNtree_lem_phrased_as_counting}, which proves the
induction step. \end{proof}

%%%%%%%%%%%%%%%%%%%%%%%%%%%%%%%%%
\subsection{Application to nodal count}

The interlacing results of the previous section allow one to prove that bi-dendral graphs have anomalous nodal counts: the nodal surplus of a bi-dendral graph with $d$ middle edges is never $\beta=d-1$ or $0$ (with the exception of
$\sigma_{1}$ which is $0$ for all graphs with Neumann conditions). We
will prove this fact for a slightly more general setup.

We consider the bi-dendral graph, some of whose middle edges can have an
intermediate point $v$ at which the anti-Neumann conditions
\eqref{eq:antiN_def} are enforced.  We can still ask the questions
about the number of zeros of the eigenfunction.  Note that there is
usually a sign change at the above $v$ but \emph{not} a zero.  It can
also happen that
$\left.f\right|_{e_{-}}(v)=\left.f\right|_{e_{+}}(v)=0$, but then, on
the contrary, $f$ does not change sign at $v$.  The $n$-th \emph{nodal
  surplus}, as before, is defined as the number of zeros of the $n$-th
eigenfunction minus $n-1$.

Placing an anti-Neumann condition at a point on the edge $e$ is unitarily
equivalent to placing a magnetic potential that integrates to $\pi$
on the edge $e$. In particular, without loss of generality we can
impose each anti-Neumann condition in the middle of the respective
edge.

We are now ready to prove theorem~\ref{thm:mandarin_mag_surplus_bounds},
which asserts that a generic bi-dendral graph
satisfies
\begin{equation*}
1\leq\sigma_{n}\leq\beta-1=d-2,\qquad n>1.
\end{equation*}

\begin{proof}[Proof of theorem~\ref{thm:mandarin_mag_surplus_bounds}]
  Let $a$ be the number of edges on which anti-Neumann conditions are
  enforced. Since the anti-Neumann conditions are symmetric and can be
  imposed in the middle of edges (their location along the edge does
  not affect the eigenvalues nor the nodal count), the graph still has
  up-down reflection symmetry (see
  figure~\ref{fig:mandarin_split_mixed}, left). Therefore, the
  spectrum can be decomposed into two halves, even and odd, consisting
  of eigenvalues whose eigenfunctions are symmetric and those with
  anti-symmetric eigenfunctions with respect to the reflection.

  \begin{figure}[t]
    \centering
    \includegraphics{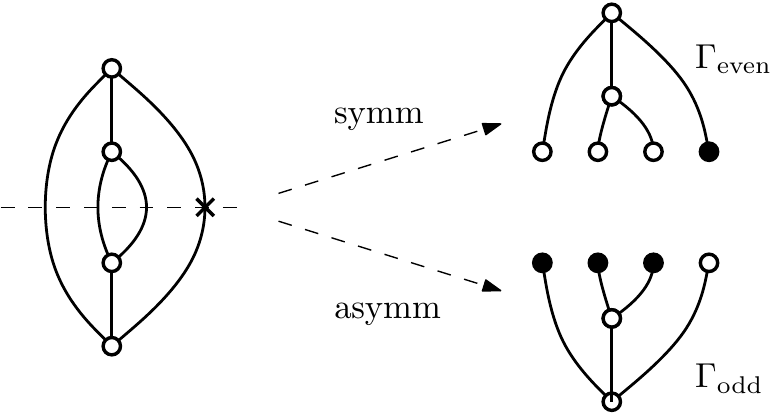}
    \caption{A bi-dendral graph with an
      anti-Neumann condition on the fourth middle edge, indicated by a
      cross. The symmetric eigenfunctions are eigenfunctions of the
      top tree graph on the right; the antisymmetric ones are the
      eigenfunctions of the bottom tree graph. Neumann conditions are
      shown as empty circles and Dirichlet conditions as filled
      circles.}
    \label{fig:mandarin_split_mixed}
  \end{figure}

  The symmetric eigenfunctions have Dirichlet conditions at the
  midpoints of anti-Neumann edges and zero derivative at the midpoints
  of the rest of the middle edges. There is a one-to-one
  correspondence between the symmetric eigenfunctions and the
  eigenfunctions of the upper half of the bi-dendral graph,
  $\Gamma_{\mathrm{even}}$, with the corresponding boundary conditions
  (see figure~\ref{fig:mandarin_split_mixed}, right). The
  antisymmetric eigenfunctions satisfy the opposite conditions and are
  in correspondence with the eigenfunction of the tree graph
  $\Gamma_{\mathrm{odd}}$.  Namely, if $\Gamma$ is a bi-dendral
  graph, then we have the spectral decomposition
  $\textrm{Spec}\left(\Gamma\right) =
  \textrm{Spec}\left(\Gamma_{\mathrm{even}}\right) \cup
  \textrm{Spec}\left(\Gamma_{\mathrm{odd}}\right)$.  The resulting
  pair of trees have opposite vertex conditions on the leaves
  and therefore fit the description in lemma~\ref{lem:DNtree_mixed}.

  Consider the eigenfunction corresponding to
  $\lambda_{k}(\Gamma_{\mathrm{even}})$, and assume that this
  eigenvalue is generic in the spectrum of $\Gamma$. As
  $\Gamma_{\mathrm{even}}$ is a tree graph, this eigenfunction has
  $k-1$ internal zeros (see Remark~\ref{rem:nodal_mag},
  part~\ref{itm:tree_count}) and $a$ zeros at the Dirichlet leaves.
  Its unfolding to the graph $\Gamma$ therefore has $2(k-1)+a$ zeros.
  %   Note that if $k=1$ and $a=0$ then
  %   $\lambda_{k}\left(\Gamma_{\mathrm{even}}\right)=0$, and this is also
  %   the first eigenvalue of $\Gamma$,
  %   $\lambda_{1}\left(\Gamma\right)=0$, with no magnetic fluxes. We
  %   assume that our eigenvalue is different and indeed for all other
  %   cases ($k>1$ or $a>0$) we may apply lemma~\ref{lem:DNtree_mixed}.
  Applying lemma~\ref{lem:DNtree_mixed} twice for comparison with the
  graph $\Gamma_{\mathrm{odd}}$ (which has $t=d-a$ Dirichlet leaves),
  we get
  \begin{equation*}
    \lambda_{k-d+a+1}(\Gamma_{\mathrm{odd}})
    < \lambda_{k}(\Gamma_{\mathrm{even}})
    < \lambda_{k+a-1}(\Gamma_{\mathrm{odd}}),
  \end{equation*}
  the second inequality being valid only when $k+a-1>0$, which is true
  due to our assumptions ($k>1$ or $a>0$).  The inequalities above are
  strict as we assumed $\lambda_{k}(\Gamma_{\mathrm{even}})$ to be
  simple when considered as an eigenvalue in the spectrum of
  $\Gamma$. The position of
  the eigenvalue $\lambda_{k}(\Gamma_{\mathrm{even}})$ in the spectrum
  of $\Gamma$ is equal to the number of eigenvalues of either $\Gamma_{\mathrm{even}}$ or $\Gamma_{\mathrm{odd}}$ that are smaller than or equal to $\lambda_{k}(\Gamma_{\mathrm{even}})$. From the inequalities above, we deduce that this position is between $2k-d+a+1$ and $2k-2+a$. The nodal surplus is
  then between
  \begin{equation*}
    2(k-1)+a-(2k-2+a-1)=1
  \end{equation*}
  and
  \begin{equation*}
    2(k-1)+a-(2k-d+a+1-1)=d-2.
  \end{equation*}

  The calculation for the antisymmetric eigenfunctions is identical
  modulo the change of $a$ to $d-a$, which has no effect on the final result.
\end{proof}

In the two simplest cases of mandarin graphs (we assume no anti-Neumann
conditions), we know all nodal counts exactly.

\begin{corollary}
  The nodal surplus count of a mandarin with 3 edges and no magnetic
  fluxes is $\{0,\ 1,\ 1,\ 1,\ldots\}$. The nodal surplus count of a
  mandarin with 4 edges and no magnetic fluxes is $\{0,\ 1,\ 2,\ 1,\
  2,\ 1,\ldots\}$.
\end{corollary}

\begin{proof}
  The nodal surplus of the first eigenfunction is always $0$ (it does
  not change sign). The nodal count of 3-mandarin follows immediately
  from theorem~\ref{thm:mandarin_mag_surplus_bounds} with $d=3$.

  Consider now the 4-mandarin graph.  The number of sign changes along
  \emph{every} edge of the graph must have the same parity: the parity
  is fully determined by the relative signs of the eigenfunction on
  the two vertices of the graph.  In either case the total number of
  zeros is even. From theorem~\ref{thm:mandarin_mag_surplus_bounds},
  the (non-trivial) nodal surplus of the 4-mandarin is either $1$ or
  $2$, but the even number of zeros forces the alternating pattern.
\end{proof}

%%%%%%%%%%%%%%%%%%%%%%%%%%%%%%%%%%%%%%%%%%%%%%%%%%%%%%%%%%
%%%%%%%%%%%%%%%%%%%%%%%%%%%%%%%%%%%%%%%%%%%%%%%%%%%%%%%%%%
\section{Critical points of the dispersion relation of infinite
  quantum graphs\label{sec:critical_points_periodic_graphs}}

%%%%%%%%%%%%%
\subsection{Dispersion relation and magnetic Laplacian}

Let $\G$ be an infinite $\Z^{k}$-periodic quantum graph with no
magnetic potential. The Schr\"odinger operator on $\G$ acts as
\begin{equation}
  \label{eq:shrod_on_qg}
  (Hf)(x)=-f''(x)+q(x)f(x),
\end{equation}
where $q(x)\in\R$ is a bounded, periodic, piecewise continuous
function that represents electric potential.  The conditions on the
vertices of $\G$ are assumed to be of $\delta$-type
~\eqref{eq:vconditions} and are also periodic.

\begin{figure}
  \centering
  \includegraphics{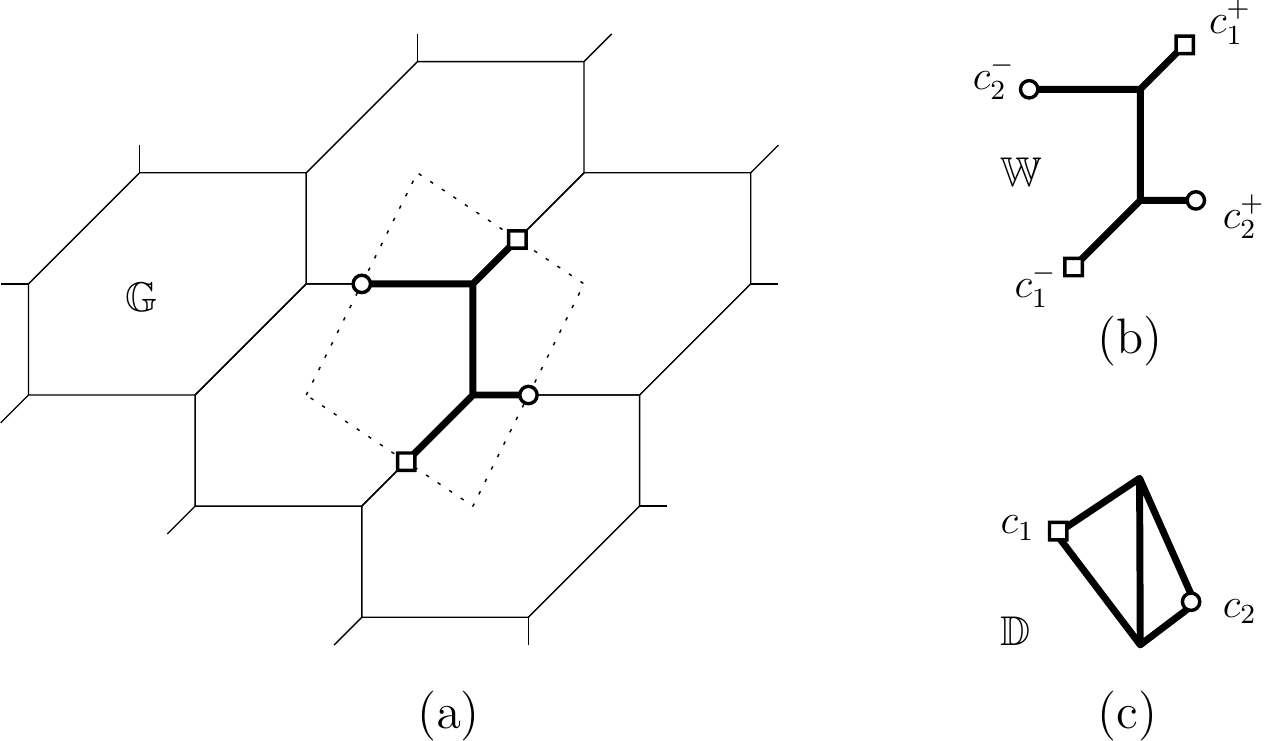}
  \caption{(a) A $\Z^{2}$-periodic graph with a chosen fundamental
    domain; (b) the chosen fundamental domain with quasi-identified vertices
    marked with the same shape; (c) the graph $\D$ formed by
    merging the quasi-identified vertices.}
  \label{fig:WtoD}
\end{figure}

Choose a fundamental domain $\W$, i.e. a connected subgraph of $\G$
which has one representative of each orbit of $\Z^{k}$ acting on $\G$
(see figure~\ref{fig:WtoD}(a) for an example).  We assume that the
boundary of $\W$ does not include any vertices of $\G$ (handling the
latter would introduce unnecessary notational difficulties).

Denote by $\left\{ s_{j}\right\} _{j=1}^{k}$ some choice of $k$
generators of $\Z^{k}$ acting on $\G$.  We call a pair of vertices, $\left(c_{j}^{+},c_{j}^{-}\right)$,
belonging to the boundary of $\W$, \emph{quasi-identified} if $s_{j}\left(c_{j}^{-}\right)=c_{j}^{+}$
(see figure~\ref{fig:WtoD}(b)).  We assume that there is only one such
pair for each $s_{j}$, $j=1,\ldots,k$.  Note that this condition may
depend on the choice of the fundamental domain.

Let $\avec=(\alpha_{1},\alpha_{2},\ldots,\alpha_{k})$ and on the graph
$\W$ define the operator
$H^{\avec}:\widetilde{H^{2}}(\W,\C)\rightarrow\widetilde{L^{2}}(\W,\C)$,
which acts as $-\frac{d^{2}}{dx^{2}}+q(x)$ on every edge, along with
the conditions
\begin{equation}
\begin{split}f(c_{j}^{-}) & =e^{i\alpha_{j}}f(c_{j}^{+}),\\
f'(c_{j}^{-}) & =-e^{i\alpha_{j}}f'(c_{j}^{+})
\end{split}
\label{eq:oldalphaconditions}
\end{equation}
at the quasi-identified vertices and the $\delta$-type vertex
conditions inherited from the graph $\G$ at all other vertices of
$\W$.

By Floquet-Bloch theory (see e.g., \cite{Kuc_floquet,BK_graphs}), we know that the
spectrum of $\G$ can be found by calculating the eigenvalues of
$H^{\avec}$ and taking the union over all $\avec$ in the
\emph{Brillouin zone}, the torus $(-\pi,\pi]^{k}$ of all possible
values of $\avec$.  In other words,
\begin{equation*}
  \sigma(H)=\bigcup_{\avec\in(-\pi,\pi]^{k}}\sigma(H^{\avec}).
\end{equation*}
The multi-valued function $\sigma(H^{\avec})$ is called the
\emph{dispersion relation}.

When the operator $H$ is real, as it is in our case, complex
conjugation transforms $H^{\avec}$ into $H^{-\avec}$, implying that
the dispersion relation $\sigma(H^{\avec})$ is symmetric with respect
to the inversion $\avec\mapsto-\avec$.  The fixed points of this
transformation are the vectors $\avec \in \{0,\pi\}^k$ with all
entries equal to either $0$ or $\pi$.  We call these vectors the
\emph{symmetry points} of the Brillouin zone.  We remark that if $\G$
has additional symmetries, a hierarchy of symmetry points may appear in
the Brillouin zone (see \cite{BerCom_prep14} for an example).
However, in this work we reserve the term for the vectors from
$\{0,\pi\}^k$ only.

The relation to a magnetic operator on a graph is explicated by the
following construction.  Beginning with the fundamental domain $\W$,
glue the vertices of each quasi-identified pair
$(c_{j}^{+},c_{j}^{-})$ to form a new vertex $c_{j}$ and denote the
resulting graph by $\D$ (see figure~\ref{fig:WtoD}(c)). When connecting
vertices, we do not change the edge lengths.  The following easy
result can be found, for example, in \cite[thm 2.6.1]{BK_graphs} or
\cite{KosSch_cmp03,Rue_prep11}.

\begin{lemma}
  \label{lemma: unit_equiv}
  The operator $H^{\avec}$ is unitarily equivalent to the operator
  $H^{A}:\widetilde{H^{2}}(\D,\C)\rightarrow\widetilde{L^{2}}(\D,\C)$,
  defined as $-\left(\frac{d}{dx}-iA(x)\right)^{2}+q(x)$ on every
  edge, with the vertex conditions
  \begin{equation*}
    \left\{ \begin{array}{l}
        g(x)\mbox{ is continuous at v},\\
        \displaystyle\sum_{e\in{E_{v}}}\left(\frac{d}{dx_e}-iA(v)\right)g(v)=\chi_{v}g(v),\qquad\chi_{v}\in\R,
      \end{array}\right.
  \end{equation*}
  where $A(x)$ is a one-form on $\D$ that satisfies
  \begin{equation*}
    \alpha_{j}=\int_{c_{j}^{-}}^{c_{j}^{+}}A(x)\mod2\pi
  \end{equation*}
  for any path on $\W$ between $c_{j}^{-}$ and $c_{j}^{+}$.

  The unitary equivalence is as follows. Choose an arbitrary point,
  $p,$ on $\W$. If $g$ is an eigenfunction of $H^{A}$, then
  $f:=ge^{-i\xi}$ is an eigenfunction of $H^{\avec}$, where
  \begin{equation*}
    \xi(x)=\int_{p}^{x}A(x).
  \end{equation*}
\end{lemma}

\begin{remark}
  The magnetic flux $\alpha_{j}$ is path independent because there is
  no other magnetic potential on our graph. The integral of $A(x)$
  around a cycle in $\W$ (i.e., a cycle that has no magnetic
  potential) is zero. Since the sign of $A(x)$ (and $\alpha_{j}$)
  depends on direction, the integral around part of a cycle in one
  direction is equal to the integral around the rest of the cycle
  traversed in the opposite direction.  The same applies to the phase
  $\xi(x)$.
\end{remark}

Lemma~\ref{lemma: unit_equiv}, in particular, shows that the
eigenvalues of $H^\avec$ do not depend on the local changes to the
choice of the fundamental domain $\W$.  Since $|f|=|g|$ and hence
$f(x)=0$ if and only if $g(x)=0$, we can assume without loss of
generality that eigenfunctions are non-zero at $c_{j}$, as well as at $c_{j}^{\pm}$.

Now we define another operator based on the graph $\W$ by specifying
different conditions at the quasi-identified vertices.  Let
$\boldsymbol{\gamma}=(\gamma_{1},\gamma_{2},\ldots,\gamma_{k})$ and
define the operator
$H^{\gvec}:\widetilde{H^{2}}(\W,\C)\rightarrow\widetilde{L^{2}}(\W,\C)$,
which acts as $-\frac{d^{2}}{dx^{2}}+q(x)$ on every edge, along with
the Robin vertex conditions
\begin{equation*}
\begin{split}f'(c_{j}^{+}) & =\gamma_{j}f(c_{j}^{+}),\\
f'(c_{j}^{-}) & =-\gamma_{j}f(c_{j}^{-})
\end{split}
\end{equation*}
at the quasi-identified vertices and the same conditions as $\G$ at all
other vertices.  It turns out that extremal points of
$\lambda_n(\avec)$ manifest themselves as multiple eigenvalues of $H^{\gvec}$.

\begin{theorem}
  \label{thm:quantumdegeneracy}
  Suppose the infinite periodic quantum graph $\G$ has no magnetic
  potential and the fundamental domain $\W$ has only one
  quasi-identified vertex pair in each direction.  Let
  $\lambda_{n}{(\avec)}$ have a critical point $\avecc$ that is not at
  a symmetry point of the Brillouin zone (i.e., $\exists j$ such that
  $\alpha_{j}^{*}\neq0,\pi$), and suppose that the eigenvalue
  $\lambda_{n}(\avecc)$ is simple with corresponding eigenfunction $f$.

  Then
  \begin{enumerate}
  \item $\gamma_{j}^{*}=\frac{f'(c_{j}^{+})}{f(c_{j}^{+})}$ is real
    for all $j$,
  \item $\lambda=\lambda_{n}(\avecc)$ is a degenerate eigenvalue of
    $H^{\boldsymbol{\gamma^{*}}}$, and
  \item if, additionally, $\W$ is a tree, there exists an internal
    vertex of $\W$, of degree three or higher, such that the
    eigenfunction $f$ is zero at this vertex. \label{enu:quantumdegeneracy-c}
  \end{enumerate}
\end{theorem}

To prove theorem~\ref{thm:quantumdegeneracy} we first collect some
auxiliary useful facts.

%%%%%%%%%%%%%%
\subsection{Critical points of the dispersion relation}
\label{sec:critical_degeneracy}

\begin{lemma}
  \label{lemma: constant}
  If $f$ is an eigenfunction of the operator \eqref{eq:shrod_on_qg}, then
  $\im(f'(x)\overline{f(x)})$ is constant on each edge of the graph.
\end{lemma}

\begin{proof}
  We start by calculating the Wronskian of the functions $f$ and $\cc{f}$:
  \begin{align*}
    W(f,\overline{f}) & =f'(x)\overline{f(x)}-f(x)\overline{f'(x)}\\
    & =f'(x)\overline{f(x)}-\overline{\overline{f(x)}f'(x)}
    = 2i\im(f'(x)\overline{f(x)}).
  \end{align*}
  On the other hand, both $f$ and $\cc{f}$ are solutions to the
  differential equation $-y''(x)+(q(x)-\lambda)y(x)=0$ on any edge
  and, by Abel's Theorem, their Wronskian is constant.
\end{proof}

We note that the value of the Wronskian changes from one edge to
another.  If all vertex conditions are Neumann, the Wronskian defines
a flow on the graph.  This and other facts about the Wronskian on
graphs can be found in \cite{BerWey_ptrsa13, Wey_phd14}.

In what follows, we will need to differentiate the function
$\lambda_{n}(\avec)$.  This is allowed since $\lambda_{n}(\avecc)$ is
a simple eigenvalue of $H^{\avecc}$, and the function
$\lambda_{n}(\avec)$ is analytic around
$\avec=\avecc$~\cite{Kato_book} (see also
\cite{BerKuc_incol12,BK_graphs} for a discussion of this fact for
quantum graphs).

\begin{lemma}
  \label{lemma: real4}
  Suppose that the infinite periodic quantum graph $\G$ satisfies the
  conditions of theorem~\ref{thm:quantumdegeneracy}.  If $\avecc$ is a
  critical point of $\lambda_{n}(\avec)$ and $\lambda_{n}(\avecc)$ is
  a simple eigenvalue, then the eigenfunction $f$ of $H^\avec$
  corresponding to $\lambda_{n}(\avecc)$ satisfies
  \begin{equation}
    f'(c_{j}^{+})\overline{f(c_{j}^{+})}\in\R
   \label{eq: alphareal}
  \end{equation}
   for all $j=1,2,\ldots,k$.
\end{lemma}

\begin{proof}
  We will show that
  \begin{equation}
    \label{eq:diff_eig_alpha}
    \frac{\partial \lambda}{\partial \alpha_j} = -2 \im\left(f'(c_{j}^{+})\overline{f(c_{j}^{+})}\right),
  \end{equation}
  which directly implies \eqref{eq: alphareal} since $\avecc$ is a
  critical point.

  By lemma~\ref{lemma: unit_equiv}, $H^\avec$ and $H^A$ have the same eigenvalues so
  \begin{equation*}
  \left.\frac{\partial \lambda(\avec)}{\partial \alpha_j} \right|_{\avec = \avecc} = \left.\frac{d}{dt} \lambda(\avecc+t\boldsymbol{\delta}\avec_j) \right|_{t=0}
    =\left.\frac{d}{dt} \lambda(A^*+tB) \right|_{t=0}
  \end{equation*}
  where $\boldsymbol{\delta}\avec_j = (0,  \ldots, \delta\alpha_j,  \ldots , 0)$ and $B(x)$ is any continuous function that satisfies
  \begin{equation*}
    \int_{c_{k}^{-}}^{c_{k}^{+}}B(x) \mod 2\pi
    = \left\{
\begin{array}{ll}
     \delta\alpha_j & k = j  \\
     0 & k \neq j
\end{array}\right. .
  \end{equation*}
  Denote by $e_{j}$ the edge of $\D$ which contains $c_{j}$.
In particular, we choose $\delta\alpha_j = 1$ and a function $B_j(x)$ that is
  compactly supported on edge $e_j$ near the point $c_j$ and satisfies
  \begin{equation}
    \label{eq:choice_of_B}
    \int_{c_{k}^{-}}^{c_{k}^{+}}B_j(x) \mod 2\pi
    = \delta_{kj}
  \end{equation}
where $\delta_{kj}$ is the Kronecker delta function.

  Let $g_{t}$ be an eigenfunction of norm one corresponding to
  $\lambda(A^{*}+tB_j)$.  We define $g:=g_{0}$ and get from the
  Hellmann-Feynman Theorem~\cite{Feynman} that
  \begin{equation*}
    \left.\frac{d}{d t}\lambda(A^*+tB_j)\right|_{t=0}
    = \left. \left( \frac{d}{dt} H^{A^*+tB_j}\right|_{t=0} g,g\right).
  \end{equation*}

  One can calculate that
  \begin{align*}
    \left.\frac{d}{dt}H^{A^{*}+t{B_{j}}}\right|_{t=0}
    & = \left.-\frac{d}{dt}\left(\frac{d}{dx}-i(A^*(x)+tB_{j}(x))\right)^{2}\right|_{t=0}\\
    & = 2iB_{j}(x)\frac{d}{dx} + iB_{j}'(x) + 2A^*(x)B_{j}(x).
  \end{align*}
  Since $B_j(x)$ is supported on the edge $e_j$ only, we get
  \begin{align*}
    \frac{\partial \lambda(\avec)}{\partial \alpha_j}
    & = \int_{e_j}iB_j'(x)g(x)\overline{g(x)}\,\mathrm{d}x
    + \int_{e_j}\left(2iB_j(x)g'(x)\overline{g(x)}
      + 2A^*(x)B_j(x)|g(x)|^2\right)\,\mathrm{d}x\\
    & = \int_{e_j}\left(-iB_j(x)\frac{d}{dx}
      \left(g(x)\overline{g(x)}\right)
      + 2iB_j(x)g'(x)\overline{g(x)} + 2A^*(x)B_j(x)|g(x)|^2\right)\,\mathrm{d}x,
  \end{align*}
  using integration by parts (the boundary terms disappear due to the
  support of $B_j(x)$).  Continuing, this gives us
  \begin{align*}
    \frac{\partial \lambda(\avec)}{\partial \alpha_j}
    &= \int_{e_j} B_j(x) \left[-ig(x)\overline{g'(x)} +
      A^*(x)g(x)\overline{g(x)} + ig'(x)\overline{g(x)} + A^*(x)g(x)\overline{g(x)}\right]\,\mathrm{d}x\\
    & =\int_{e_j} 2 B_j(x)\im\left[\left(-g'(x)+iA^*(x)g(x)\right) \overline{g(x)}\right]\,\mathrm{d}x,
  \end{align*}
  where moving from the first line to the second we use the fact
  that $A^*(x)$ is real. By lemma~\ref{lemma: unit_equiv}, we know that
  since $g$ is an eigenfunction corresponding to $\lambda(A^*)$,
  $f=ge^{-i\xi}$ is an eigenfunction corresponding to
  $\lambda(\avecc)$, and therefore
  \begin{equation*}
    \im\left[\left(g'(x)-iA^*(x)g(x)\right)\overline{g(x)}\right]
    = \im\left(f'(x)\overline{f(x)}\right).
  \end{equation*}
  However, by lemma~\ref{lemma: constant} the latter value is a constant on
  the edge $e_j$, and using \eqref{eq:choice_of_B} we get
  \begin{equation*}
    \frac{\partial \lambda(\avec)}{\partial \alpha_j}
    = -2 \im\left(f'(x)\overline{f(x)}\right) \int_{e_j} B_j(x)
    \,\mathrm{d}x
    = -2 \im\left(f'(x)\overline{f(x)}\right).
  \end{equation*}
\end{proof}

We are now ready to prove the main theorem of the previous subsection.

\begin{proof}[Proof of theorem~\ref{thm:quantumdegeneracy}]
  Let $f$ be an eigenfunction of $H^{\avecc}$ (with $\avecc$ not at a
  symmetry point of the Brillouin zone), and let
  \begin{equation*}
    \gamma_{j}^{*}=\frac{f'(c_{j}^{+})}{f(c_{j}^{+})}, \qquad j=1,\ldots,k.
  \end{equation*}
  According to lemma~\ref{lemma: real4}, all $\gamma_{j}^{*}$'s are real, and
  consequently the operator $H^{\gvecc}$ is self-adjoint and real.

  It is easy to see that $f$ is an eigenfunction of $H^{\gvecc}$.
  However, $f$ cannot be made real since it satisfies~\eqref{eq:oldalphaconditions} and we assumed that there exists $j$ with
  $\alphac_{j}\neq0,\pi$.  This is not a contradiction only if the
  real and imaginary parts of $f$ are both eigenfunctions of
  $H^{\gvecc}$, in which case $\lambda$ must be a degenerate
  eigenvalue of $H^{\gvecc}$.

  Furthermore, if $\W$ is a tree, we can apply
  \cite[cor~3.1.9]{BK_graphs}, which says that if an eigenvalue
  of a tree is degenerate there exists an internal vertex of degree
  three or higher at which all eigenfunctions from the eigenspace
  vanish.
\end{proof}

Theorem~\ref{thm:magnetic_cp_condition} now follows.
 \begin{proof}[Proof of theorem~\ref{thm:magnetic_cp_condition}]
 Let $\Gamma$ be a graph with first Betti number equal to $\beta$ and denote by $\avec\in(-\pi,\pi]^{\beta}$ the total fluxes through some choice of cycles of the graph that form a
basis of its fundamental group. Therefore, one may cut the graph at $\beta$ positions to make it a tree graph. The obtained tree graph serves as a fundamental domain $\W$ (with respect to translations) of an infinite $\Z^{\beta}$-periodic quantum graph, $\G$, as in theorem~\ref{thm:quantumdegeneracy}. The statement in theorem~\ref{thm:magnetic_cp_condition} now follows from theorem~\ref{thm:quantumdegeneracy}\eqref{enu:quantumdegeneracy-c}, realizing that the quasi-momenta of $\G$ are exactly the magnetic fluxes of $\Gamma$.
 \end{proof}

%%%%%%%%%
\subsection{Zeros and touching bands for mandarin graphs}

In light of theorem~\ref{thm:quantumdegeneracy}, we observe that a
special role is played by eigenfunctions that are zero on at least one
vertex. We now apply this observation to mandarin graphs.

\begin{proof}[Proof of theorem~\ref{thm:mandarin_magnetic_cp}]
  We will now show that if we have an extremum in the dispersion
  relation of the mandarin graph, it is due to touching bands.

  Assume the contrary: an extremum of $\lambda_{n}(\avec)$ occurs at
  $\avec=\avecc$ and $\lambda_{n}(\avecc)$ is simple. The minimum of
  $\lambda_{1}$ always happens at the point $\avec=(0,0,\ldots)$ and
  we exclude it from further considerations.

  First we argue that the eigenfunction $f$ corresponding to
  $\lambda_{n}(\avecc)$ must vanish at a vertex of the graph.  Indeed,
  if $\avecc$ is not a symmetry point, the claim follows directly from
  theorem~\ref{thm:magnetic_cp_condition}.  On the other hand, if
  $\avecc$ is a symmetry point and is non-vanishing on the vertices,
  by combining theorems~\ref{thm:mandarin_mag_surplus_bounds} and
  \ref{thm:nodal_mag}, we conclude that $\avecc$ is a saddle point,
  which contradicts it being an extremum.

  Now, without loss of generality, assume that $f$ vanishes on the top
  vertex of the $d$-mandarin graph. In addition, $f$ satisfies Neumann conditions there.  We can shift the magnetic condition to
the top vertex, resulting in
\begin{align}
 & f_{1}(v)=f_{2}(v)=\ldots=f_{d}(v)=0,\label{eq:top_Dirichlet}\\
 & e^{i\alpha_{1}}f_{1}'(v)+\ldots+e^{i\alpha_{d-1}}f_{d-1}'(v)+f_{d}'(v)=0,\label{eq:top_Neumann}
\end{align}
the standard Neumann condition at the bottom vertex, and a non-magnetic
operator acting on the edges. Ignoring, for a moment, condition \eqref{eq:top_Neumann},
we get a standard star graph with $d$ edges and Dirichlet conditions
at the boundary vertices. Such a graph, for generic choice of lengths, has
  eigenfunctions that do not vanish on the central vertex. Hence, we assume that $f$ is not equal to zero at the bottom vertex.

  Apply the top-down (vertical) reflection $F$ to the function $f$ (see
  figure~\ref{fig:mandarin_split_mixed} (left)), followed by complex
  conjugation. It is immediate to check that the new function, $\cc{Ff}$,
  satisfies the same eigenvalue problem with
  $\avec=\avecc$ as the function $f$. It is, however,
  a different function: $f$ vanishes on the top vertex and \emph{does
    not} vanish on the bottom; the function $\cc{Ff}$ does the
  opposite. We conclude that $\lambda_{n}(\avecc)$ is a multiple
  eigenvalue.
\end{proof}

%%%%%%%%%%%%%%%%%%%%%%%%%%%%%%%%%%%%%%%%%%%%%%%%%
%%%%%%%%%%%%%%%%%%%%%%%%%%%%%%%%%%%%%%%%%%%%%%%%%
\appendix
\section{Discrete Graphs}
\label{sec:discrete_graphs}

In this section we consider the analogues for discrete graphs of some
of the theorems proved above, namely
theorems~\ref{thm:mandarin_mag_surplus_bounds},
\ref{thm:magnetic_cp_condition}, and \ref{thm:mandarin_magnetic_cp}.

As for the mandarin graphs, one can consider their discrete analogues
by placing several intermediate degree-two vertices per edge.
However, the anomaly of the nodal count
(theorem~\ref{thm:mandarin_mag_surplus_bounds}) is only partially
exhibited in this case. More specifically, in numerical experiments we
saw that the nodal surplus stays anomalous (i.e., $\sigma_{n}\neq0$
and $\sigma_{n}\neq\beta$) only in the bottom half of the
spectrum. Increasing the number of intermediate points per edge, one
can approximate any given eigenfunction of the quantum graph, but with
a discrete eigenfunction that stays ``low'' in the spectrum, so these
observations do not contradict
theorem~\ref{thm:mandarin_mag_surplus_bounds}.

We introduce the relevant definitions for discrete graphs in section
\ref{sub:introduction_to_discrete} and discuss the extrema of
dispersion relations of infinite periodic discrete graphs in
section~\ref{sub:Periodic-discrete-graphs}.

%%%%%%%%%
\subsection{Introduction to discrete graphs\label{sub:introduction_to_discrete}}

Let $\Gamma=(V,E)$ be a simple connected finite graph with vertex
set $V$ and edge set $E$. We define the Schr\"odinger operator
with the potential $q:V\to\R$ by
\begin{equation}
  H:\C^{|V|}\to\C^{|V|},
  \qquad
  \left(Hf\right)\left(u\right)
  = -\sum_{v\sim u}f\left(v\right)+q\left(u\right)f\left(u\right).
  \label{eq:discr_schrod}
\end{equation}
That is, the matrix $H$ is
\begin{equation}
  H=Q-C,
  \label{eq:matrix_H}
\end{equation}
where $Q$ is the diagonal matrix of site potentials $q\left(u\right)$
and $C$ is the adjacency matrix of the graph. It is perhaps more usual
(and physically motivated) to represent the Hamiltonian as $H=Q+L$,
where the Laplacian $L$ is given by $L=D-C$ with $D$ being the
diagonal matrix of vertex degrees.  Since we will not be imposing
any restrictions on the potential $Q$, we absorb the matrix $D$ into
$Q$. The operator $H$ has $|V|$ eigenvalues, which we number in
increasing order as before.

The nodal count of a (real) eigenfunction $f_{n}$ is defined as the
number of edges on which the eigenfunction changes sign, i.e.,
\begin{equation*}
  \phi_{n}=\Big|\{(u,v)\in E\,:\, f_{n}(u)f_{n}(v)<0\}\Big|.
\end{equation*}
This count is most relevant for eigenfunctions which do not vanish at
vertices. Otherwise, there exists alternative definitions
\cite{BiyLeySta_book,BanOreSmi_incoll08}, but they are not relevant
for the results in this paper.

We define the magnetic Hamiltonian (magnetic Schr\"odinger operator)
on discrete graphs as
\begin{equation}
  \left(H^Af\right)\left(u\right)
  = -\sum_{v\sim u}e^{iA_{v,u}}f\left(v\right) +
  q\left(u\right)f\left(u\right),
  \label{eq:discr_mag_schrod}
\end{equation}
with the convention that $A_{v,u}=-A_{u,v}$, which makes $H^A$ self-adjoint.
For further details, the reader should consult
\cite{LieLos_dmj93,Sun_cm94,CdV_spectre,CdVToHTru_afstm11}.

A sequence of vertices $C=[u_{1},u_{2},\ldots,u_{n}]$ is called a
\emph{cycle} if each two consecutive vertices, $u_{j} \mbox{ and } u_{j+1}$, are connected
by an edge ($u_{n+1}$ is understood as $u_{1}$). The flux through
the cycle $C$ is defined as
\begin{equation}
  \Phi_{C} =
  \left(A_{u_1,u_2} + \ldots + A_{u_{n-1},u_n} +
    A_{u_n,u_1}\right)\mod 2\pi.
  \label{eq:flux_through_C}
\end{equation}
Two operators which have the same flux through every cycle $C$ are
unitarily equivalent (by a gauge transformation). Therefore, the effect
of the magnetic field on the spectrum is fully determined by $\beta$
fluxes through a chosen set of basis cycles of the cycle space. We
denote them by $\alpha_{1},\ldots,\alpha_{\beta}$ and consider the
$n$-th eigenvalue of the graph as a function of $\avec$.

A result similar to theorem~\ref{thm:nodal_mag} holds for
discrete graphs (in fact, discrete graphs were the original context
in which the magnetic-nodal connection was established \cite{Ber_apde13,CdV_apde13}).
Note that in the discrete version of the theorem the nodal count should
be modified if one considers a symmetric point $\avecc\in\left\{ 0,\pi\right\} ^{\beta}$
which differs from zero. The nodal count one should consider in such
a case is

\begin{equation*}
  \phi_{n}=\Big|\{(u,v)\in E\,:\, H_{u,v}\left(\avecc\right)f_{n}(u)f_{n}(v)>0\}\Big|,
\end{equation*}
and its surplus, $\phi_{n}-\left(n-1\right)$, equals the Morse index
of $\lambda_{n}$ at the the symmetric point, $\avecc$, according
to the corresponding theorem \cite{Ber_apde13,CdV_apde13}. Note that
this modified nodal count is identical to the one previously defined
if $\avecc=\boldsymbol{0}$.

%%%%%%%%%%
\subsection{Periodic discrete graphs}
\label{sub:Periodic-discrete-graphs}

Let $\G$ be an infinite $\Z^{k}$-periodic discrete graph such
that its Hamiltonian has a $\Z^{k}$-periodic electric potential,
$q$, and no magnetic potential. We denote by $s_{j}(v)$ the
vertex in $\G$ resulting from shifting the vertex $v$ in the positive
$j^{th}$ $\Z^{k}$-direction. We choose a fundamental domain $\W$ and
define two vertices $(v,u)\in\W$ to be the \textit{$j^{th}$ quasi-connected
pair} if the vertex $s_{j}(v)$ is connected to $u$ in $\G$ (see
figure~\ref{fig:fundamentaldomain}). In general, a pair of quasi-connected
vertices is not connected in $\W$. We restrict ourselves to graphs $\G$
that have only one quasi-connected pair in each direction, at least
for some choice of $\W$.

\begin{figure}
  \centering
  \includegraphics[width=10cm,height=6cm]{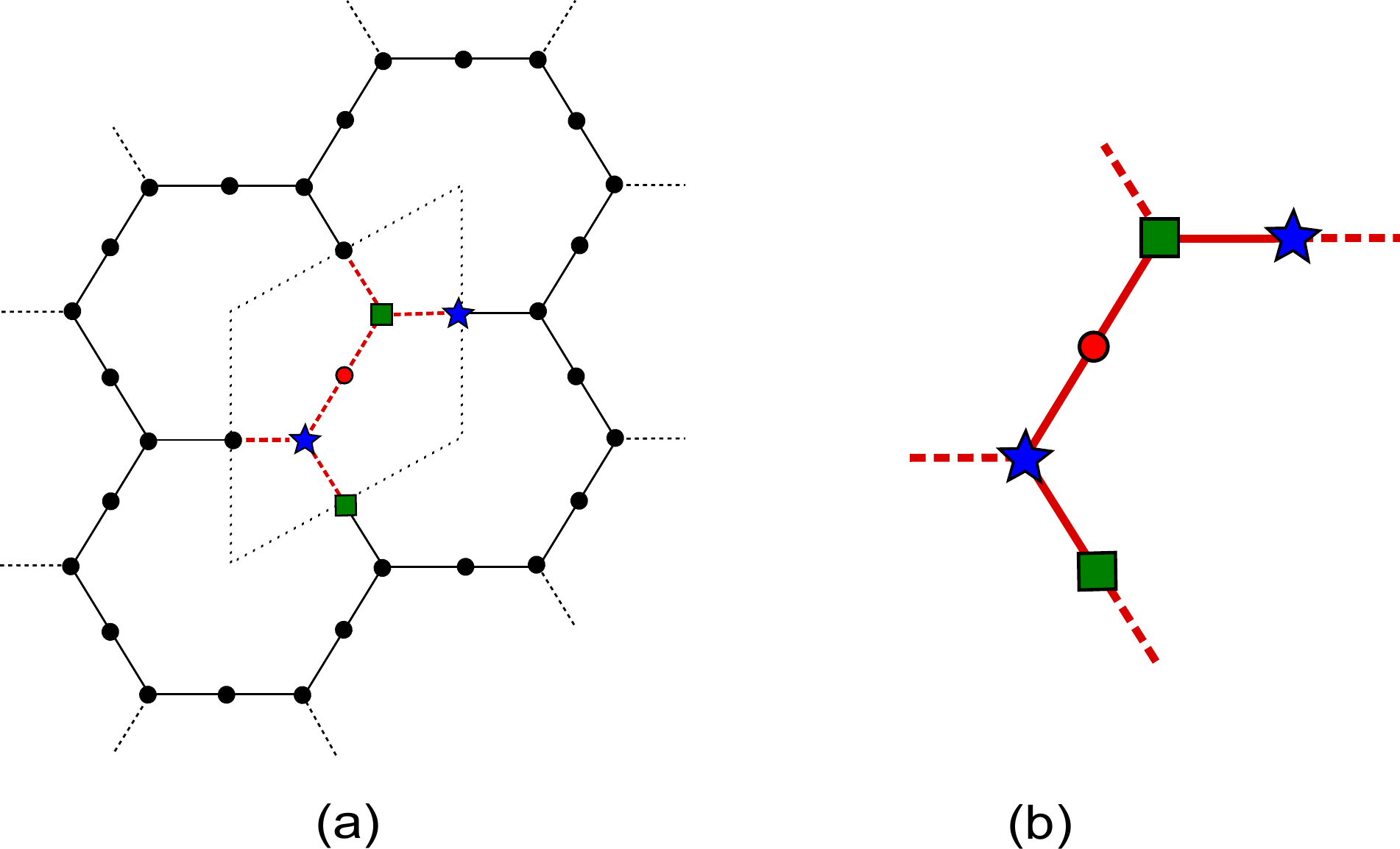}
  \caption{An infinite $\Z^{2}$-periodic graph with hexagonal lattice
    and chosen fundamental domain in box (a); the chosen fundamental
    domain with quasi-connected vertices marked with the same shape and
    color (b).}
  \label{fig:fundamentaldomain}
\end{figure}

Let $\avec=(\alpha_{1},\alpha_{2},\ldots,\alpha_{k})$ and on the
fundamental domain $\W$ define the operator $H^{\avec}=Q-C-M^{\avec}$,
where $Q$ is a diagonal matrix with real entries representing electric
potential, $C$ is the connectivity matrix of $\W$, and
\begin{equation*}
  M_{(v,u)}^{\avec} =
  \begin{cases}
    e^{i\alpha_{j}} & \mbox{if $(v,u)$ is the $j^{th}$ quasi-connected pair},\\
    e^{-i\alpha_{j}} & \mbox{if $(u,v)$ is the $j^{th}$ quasi-connected pair},\\
    0 & \mbox{otherwise}.
  \end{cases}
\end{equation*}

Similarly, let $\boldsymbol{\gamma}=(\gamma_{1},\gamma_{2},\ldots,\gamma_{k})$
and on $\W$ define the operator $H^{\gvec}=Q-C-M^{\gvec}$ where
$M^{\gvec}$ is the diagonal matrix
\begin{equation*}
  M_{(v,v)}^{\gvec} =
  \begin{cases}
    \gamma_j & \mbox{if $(v,u)$ is the $j^{th}$ quasi-connected pair},\\
    \frac{1}{\gamma_j} & \mbox{if $(u,v)$ is the $j^{th}$ quasi-connected pair},\\
    0 & \mbox{otherwise}.
  \end{cases}
\end{equation*}

We are now ready to present the main theorem of this appendix, which
is the discrete analogue of theorem~\ref{thm:quantumdegeneracy}.

\begin{theorem}
  \label{thm: discretedegeneracy}
  Suppose the infinite periodic discrete graph $\G$ has no magnetic
  potential and only one quasi-connected vertex pair in each
  direction, $\lambda_{n}(\avec)$ has a critical point $\avecc$ that
  is not at a symmetry point of the Brillouin zone (i.e., $\exists j$
  such that $\alpha_{j}^{*}\neq0,\pi$), the eigenvalue
  $\lambda_{n}(\avecc)$ is simple, and the corresponding eigenvector
  $f$ is non-zero at all quasi-connected vertex pairs.  Then
  \begin{enumerate}
  \item $\gamma_{j}^{*}=e^{i\alpha_{j}^{*}}\frac{f(u)}{f(v)}$ is real for all $j$ where
    $(v,u)$ is the $j^{th}$ quasi-connected pair,
  \item $\lambda=\lambda_{n}(\avecc)$ is a degenerate eigenvalue of
    $H^{\gvecc}$, and
  \item if, additionally, $\W$ is a tree, there exists an internal
    vertex of $\W$, of degree three or higher, such that the
    eigenfunction $f$ is zero at this vertex.
  \end{enumerate}
\end{theorem}

In order to prove theorem~\ref{thm: discretedegeneracy},
we establish the following two lemmas.

\begin{lemma}
  \label{lemma: real}
  Suppose that the discrete infinite periodic graph $\G$ has no
  magnetic potential and only one quasi-connected vertex pair in each
  direction. If $\avecc$ is a critical point of $\lambda_{n}(\avec)$ and
  $\lambda_{n}(\avecc)$ is a simple eigenvalue, then the eigenvector
  $f$ corresponding to $\lambda_{n}(\avecc)$ satisfies
  \begin{equation}
    e^{i\alpha_{j}^{*}}f(u)\overline{f(v)}\in\R
    \qquad \forall j=1,2,\ldots,k
    \label{eq: real lemma}
  \end{equation}
  where $(v,u)$ is the $j^{th}$ quasi-connected pair.
\end{lemma}

\begin{proof} Since $\lambda_{n}(\avecc)$ is simple,
  $\lambda_{n}(\avec)$ is analytic around the critical point
  $\avec=\avecc$.  Let $f_{\avec}$ be an eigenvector of norm one
  corresponding to $\lambda_{n}(\avec)$.  In particular,
  $f=f_{\avecc}$.  We have
  \begin{equation*}
    \left(\frac{\partial H^{\avec}}{\partial\alpha_{j}}\right)_{(v,u)}
    = -\left(\frac{\partial
        M^{\avec}}{\partial\alpha_{j}}\right)_{(v,u)}
    =
    \begin{cases}
      -ie^{i\alpha_{j}} & \mbox{if $(v,u)$ is the $j^{th}$ quasi-connected pair},\\
      ie^{-i\alpha_{j}} & \mbox{if $(u,v)$ is the $j^{th}$ quasi-connected pair},\\
      0 & \mbox{otherwise.}
    \end{cases}
  \end{equation*}
  Therefore, since $\avecc$ is a critical point of $\lambda_{n}(\avec)$,
  one can see that
  \begin{align*}
    0 = \left. \frac{\partial}{\partial\alpha_j} \lambda_n(\avec)
    \right|_{\avec=\avecc} = \left(\frac{\partial
        H^\avecc}{\partial\alpha_j}f, f\right)
    & =  -i e^{i\alpha_j^*} f(u)\cc{f(v)} + i e^{-i\alpha_j^*} f(v)\cc{f(u)}\\
    & = 2\im\left[e^{i\alpha_j^*}f(u)\cc{f(v)}\right],
  \end{align*}
  which completes the proof.
\end{proof}

\begin{lemma}
  \label{lemma: eigenvector}
  Suppose that $\G$ is a discrete infinite periodic graph and $\avecc$
  is a critical point of $\lambda_{n}(\avec)$. If the eigenvector $f$
  of $H^{\avecc}$ corresponding to $\lambda=\lambda_{n}(\avecc)$ is
  non-zero at all quasi-connected vertex pairs, then $f$ is also an
  eigenvector of $H^{\gvecc}$ corresponding to the same eigenvalue
  $\lambda$ where $(v,u)$ is the $j^{th}$ quasi-connected vertex pair
  and
  \begin{equation*}
    \gamma_{j}^{*}=e^{i\alpha_{j}^{*}}\frac{f(u)}{f(v)} \in \R.
  \end{equation*}
\end{lemma}

\begin{proof}
  We will demonstrate that $H^{\gvecc}f=H^{\avecc}f=\lambda f$.  Using
  the definitions of the operators, one can see that this is
  equivalent to showing
  \begin{equation*}
    H^{\gvecc}f=(Q-C)f-M^{\gvecc}f=(Q-C)f-M^{\avecc}f=H^{\avecc}f=\lambda f,
  \end{equation*}
  or in other words
  \begin{equation}
    M^{\gvecc}f=M^{\avecc}f.
    \label{eq: Mgamma}
  \end{equation}

  Suppose that $(v,u)$ is the $j^{th}$ quasi-connected vertex pair.
  Then at $v$ we have
  \begin{align*}
    (M^{\gvecc}f)(v) =\gamma_{j}^{*}f(v)
    = e^{i\alpha_{j}^{*}}\frac{f(u)}{f(v)}f(v) = e^{i\alphac_{j}}f(u)
    = (M^{\avecc}f)(v).
  \end{align*}
  Similarly, at $u$ we have
  \begin{align*}
    (M^{\gvecc}f)(u) =\frac{1}{\gamma_{j}^{*}}f(u)
    = e^{-i\alphac_{j}}f(v) = (M^{\avecc}f)(u),
  \end{align*}
  and at a vertex $w$ which is not in a quasi-connected pair,
  \begin{equation*}
    (M^{\gvecc}f)(w)=0=(M^{\avecc}f)(w).
  \end{equation*}
\end{proof}

\begin{proof}[Proof of theorem~\ref{thm: discretedegeneracy}]
  By lemma~\ref{lemma: eigenvector}, the eigenvector $f$ of
  $H^{\avecc}$ is also an eigenvector of $H^{\gvecc}$ (with the same
  eigenvalue $\lambda$). Choose $j$ such that
  $\alpha_{j}^{*}\neq0,\pi$ (such $\alpha_{j}^{*}$ exists by the
  theorem's conditions). By lemma~\ref{lemma: real},
  \begin{equation*}
    e^{i\alpha_{j}^{*}}f(u)\overline{f(v)}\in\R
  \end{equation*}
  for all $j$, which means that $f$ has some non-real
  entries and cannot be made real by scalar multiplication. However, it also means that the operator $H^{\gvecc}$ is
  real and self-adjoint. Therefore, the real and imaginary parts
  of $f$ must both be linearly independent eigenvectors of $H^{\gvecc}$, which implies that
  $\lambda$ is a degenerate eigenvalue of $H^{\gvecc}$.

  If $\W$ is a tree, a degenerate eigenvalue can only occur if there
  exists an internal vertex of degree at least three at which all
  eigenvectors from the eigenspace vanish \cite{Fie_cmj75a}.
\end{proof}

%%%%%%%%%%%%%%%%%%%%%%%%%%%%%%%%%%%%%%%%%%%%%%%%%%%
\section*{Acknowledgement}

We are grateful to Peter Kuchment for encouraging discussions and
comments and to anonymous referee for the most attentive reading of
our work and resulting numerous improvements.  GB was partially
supported by the NSF (grant no. DMS-1410657). RB was supported by the
ISF (grant no. 494/14), Marie Curie Actions (grant
no. PCIG13-GA-2013-618468), and the Taub Foundation (Taub Fellow).

\bibliographystyle{amsalpha}
\bibliography{bk_bibl,nodal_conical}

\end{document}